\documentclass{llncs}

\usepackage{macros}

\newcommand\set[1]{\{ #1 \}}
\newcommand\tuple[1]{{( #1 )}}

\newcommand{\REFlem}[1]{\text{Lem.~\ref{#1}}}
\newcommand{\REFthm}[1]{\text{Thm.~\ref{#1}}}

\newcommand{\REFalg}[1]{Alg.~\ref{#1}}

\newcommand{\REFsec}[1]{Sec.~\ref{#1}}

\newcommand{\REFfig}[1]{Fig.~\ref{#1}}
\newcommand{\REFass}[1]{Assumption~\ref{#1}}

\newcommand{\Cpre}{\mathrm{CPre}}
\newcommand{\FCpre}[1]{\Cpre_{\Sa{#1}}}

\newcommand{\Upre}[1]{\mathrm{Pre}_{#1}}
\newcommand{\FUpre}[1]{\Upre{\Sa{#1}}}
\newcommand{\JUpre}[1]{\Upre{\Aa{#1}}}
\newcommand{\OA}[1]{\Gamma^\uparrow_{#1}}
\newcommand{\UA}[1]{\Gamma^\downarrow_{#1}}

\newcommand{\ON}[1]{\operatorname{#1}}

\def\clap#1{\hbox to 0pt{\hss#1\hss}}

\makeatletter 
\newif\ifFIRST
\newif\ifSECOND
\let\LISTOP\relax
\newcommand{\List}[4][\;]{#3#1%
	\FIRSTtrue
	\@for\i:=#2\do{%
	\ifFIRST\LISTOP{\i}\FIRSTfalse\else,\LISTOP{\i}\fi%
	}%
	#1#4%
	\let\LISTOP\relax
}
\makeatother
\newcounter{DINGLIST}
\stepcounter{DINGLIST}
\newcommand{\markD}[3][\;\;]{\text{\ding{\the\numexpr171+\theDINGLIST}\stepcounter{DINGLIST}}#1#3}

\makeatletter

\newcommand{\propNeg}{\@ifstar\propNegStar\propNegNoStar}
\newcommand{\propNegStar}[1]{\ensuremath{\left(\propNegNoStar{#1}\right)}}
\newcommand{\propNegNoStar}[2][\cdot]{\ensuremath{\neg\ifthenelse{\isempty{#2}}{#1}{#2}}}

\newcommand{\propConj}{\@ifstar\propConjStar\propConjNoStar}
\newcommand{\propConjStar}[2]{\ensuremath{\left(\propConjNoStar{#1}{#2}\right)}}
\newcommand{\propConjNoStar}[3][\cdot]{\ensuremath{\ifthenelse{\isempty{#2}}{#1}{#2}\wedge\ifthenelse{\isempty{#3}}{#1}{#3}}}

\newcommand{\propDisj}{\@ifstar\propDisjStar\propDisjNoStar}
\newcommand{\propDisjStar}[2]{\ensuremath{\left(\propDisjNoStar{#1}{#2}\right)}}
\newcommand{\propDisjNoStar}[3][\cdot]{\ensuremath{\ifthenelse{\isempty{#2}}{#1}{#2}\vee\ifthenelse{\isempty{#3}}{#1}{#3}}}

\newcommand{\propImp}{\@ifstar\propImpStar\propImpNoStar}
\newcommand{\propImpStar}[2]{\ensuremath{\left(\propImpNoStar{#1}{#2}\right)}}
\newcommand{\propImpNoStar}[3][\cdot]{\ensuremath{\ifthenelse{\isempty{#2}}{#1}{#2}\Rightarrow\ifthenelse{\isempty{#3}}{#1}{#3}}}

\newcommand{\propAequ}{\@ifstar\propAequStar\propAequNoStar}
\newcommand{\propAequStar}[2]{\ensuremath{\left(\propAequNoStar{#1}{#2}\right)}}
\newcommand{\propAequNoStar}[3][\cdot]{\ensuremath{\ifthenelse{\isempty{#2}}{#1}{#2}\Leftrightarrow\ifthenelse{\isempty{#3}}{#1}{#3}}}

\newcommand{\AllQ}{\@ifstar\AllQStar\AllQNoStar}
\newcommand{\AllQStar}[3][\;]{\ensuremath{\left(\forall #2#1.#1#3\right)}}
\newcommand{\AllQNoStar}[3][\;]{\ensuremath{\forall #2#1.#1#3}}
\newcommand{\AllQu}{\@ifstar\AllQuStar\AllQuNoStar}
\newcommand{\AllQuStar}[3][\;]{\ensuremath{\left(\forall^{\infty} #2#1.#1#3\right)}}
\newcommand{\AllQuNoStar}[3][\;]{\ensuremath{\forall^{\infty} #2#1.#1#3}}

\newcommand{\ExQ}{\@ifstar\ExQStar\ExQNoStar}
\newcommand{\ExQStar}[3][\;]{\ensuremath{\left(\exists #2#1.#1#3\right)}}
\newcommand{\ExQNoStar}[3][\;]{\ensuremath{\exists #2#1.#1#3}}

\newcommand{\NExQ}{\@ifstar\NExQStar\NExQNoStar}
\newcommand{\NExQStar}[3][\;]{\ensuremath{\left(\nexists #2#1.#1#3\right)}}
\newcommand{\NExQNoStar}[3][\;]{\ensuremath{\nexists #2#1.#1#3}}

\newcommand{\UniqueQ}{\@ifstar\UniqueQStar\UniqueQNoStar}
\newcommand{\UniqueQStar}[3][\;]{\ensuremath{\left(\exists! #2#1.#1#3\right)}}
\newcommand{\UniqueQNoStar}[3][\;]{\ensuremath{\exists! #2#1.#1#3}}

\newenvironment{propConjA}{\left(\def\unionAtest{1}\begin{array}{@{\if\unionAtest1\gdef\unionAtest{0}\phantom{\wedge}\else\wedge\fi}l@{}}}{\end{array}\right)}

  \newlength{\SFS@HEIGHT}
  \newlength{\SFS@WIDTH}
  \newcommand{\SplitX}[2]{
	    \settoheight{\SFS@HEIGHT}{$#2$}
	    \settowidth{\SFS@WIDTH}{$#2$}
	    \mbox{\begin{tikzpicture}[baseline=(current bounding box.center)]
	    \node[] (E) at (0,0) {$#1$};
	    \node[inner sep=0pt] (F) at ($(E.south west)+(1ex,-1ex)+(3ex+.5\SFS@WIDTH,-\SFS@HEIGHT)$) {$#2$};
	    \node[] (E) at (0,0) {\phantom{$#1$}};
	    \draw[fill] ($(E.east)+(1ex,0ex)$) circle (.2ex);
	    \draw[-] ($(E.east)+(1ex,0ex)$) -- ($(E.south east)+(1ex,-0.5ex)$) -- ($(E.south west)+(1ex,-0.5ex)$) -- ($(E.south west)+(1ex,-1ex)-(0,\SFS@HEIGHT)$) -- ($(E.south west)+(2.5ex,-1ex)-(0,\SFS@HEIGHT)$);
	    \draw[fill] ($(E.south west)+(2.5ex,-1ex)-(0,\SFS@HEIGHT)$) circle (.2ex);
	    \end{tikzpicture}}}
  \newcommand{\SplitS}[2]{
	    \settoheight{\SFS@HEIGHT}{$#2$}
	    \settowidth{\SFS@WIDTH}{$#2$}
	    \mbox{\begin{tikzpicture}[baseline=(current bounding box.center)]
	    \node[] (E) at (0,0) {$#1$};
	    \node[inner sep=0pt] (F) at ($(E.south west)+(1ex,0.5ex)+(3ex+.5\SFS@WIDTH,-\SFS@HEIGHT)$) {$#2$};
	    \end{tikzpicture}}}

\newcommand{\semantics}[1]{\langle\![#1]\!\rangle}
\newcommand{\Set}[2][]{\List[#1]{#2}{\left\{}{\right\}}}
\newcommand{\VSet}[2][]{\let\LISTOP\val\List[#1]{#2}{\{}{\}}}

\newcommand{\VTuple}[2][]{\let\LISTOP\val\List[#1]{#2}{(}{)}}

\newcommand{\UNION}{\@ifstar\UNIONStar\UNIONNoStar}
\newcommand{\UNIONStar}[2]{\ensuremath{\left(\UNIONNoStar{#1}{#2}\right)}}
\newcommand{\UNIONNoStar}[2]{\ensuremath{\ifthenelse{\isempty{#1}}{\cdot}{#1}\cup\ifthenelse{\isempty{#2}}{\cdot}{#2}}}

\newcommand{\UNIOND}{\@ifstar\UNIONDStar\UNIONDNoStar}
\newcommand{\UNIONDStar}[2]{\ensuremath{\left(\UNIONDNoStar{#1}{#2}\right)}}
\newcommand{\UNIONDNoStar}[2]{\ensuremath{\ifthenelse{\isempty{#1}}{\cdot}{#1}\uplus\ifthenelse{\isempty{#2}}{\cdot}{#2}}}

\newcommand{\SETMINUS}{\@ifstar\SETMINUSStar\SETMINUSNoStar}
\newcommand{\SETMINUSStar}[2]{\ensuremath{\left(\SETMINUSNoStar{#1}{#2}\right)}}
\newcommand{\SETMINUSNoStar}[2]{\ensuremath{\ifthenelse{\isempty{#1}}{\cdot}{#1}\setminus\ifthenelse{\isempty{#2}}{\cdot}{#2}}}

\newcommand{\INTERSECT}{\@ifstar\INTERSECTStar\INTERSECTNoStar}
\newcommand{\INTERSECTStar}[2]{\ensuremath{\left(\INTERSECTNoStar{#1}{#2}\right)}}
\newcommand{\INTERSECTNoStar}[2]{\ensuremath{\ifthenelse{\isempty{#1}}{\cdot}{#1}\cap\ifthenelse{\isempty{#2}}{\cdot}{#2}}}

\newcommand{\CARTPROD}{\@ifstar\CARTPRODStar\CARTPRODNoStar}
\newcommand{\CARTPRODStar}[2]{\ensuremath{\left(\CARTPRODNoStar{#1}{#2}\right)}}
\newcommand{\CARTPRODNoStar}[2]{\ensuremath{\ifthenelse{\isempty{#1}}{\cdot}{#1}\times\ifthenelse{\isempty{#2}}{\cdot}{#2}}}

\newcommand{\FINCOUNT}{\@ifstar\FinCountStar\FinCountNoStar}
\newcommand{\FinCountStar}[1]{\ensuremath{\#(\ifthenelse{\isempty{#1}}{\cdot}{#1})}}
\newcommand{\FinCountNoStar}[1]{\ensuremath{\#\left(\ifthenelse{\isempty{#1}}{\cdot}{#1}\right)}}

\makeatother

\newcommand{\real}[1]{\ifstrempty{#1}{\mathbb{R}}{\mathbb{R}^{#1}}}
\newcommand{\Z}{\mathbb{Z}}
\newcommand{\N}{\mathbb{N}}

\newcommand{\fun}{\ensuremath{\ON{\rightarrow}}}

\newcommand{\SetComp}[3][]{\{#1#2#1\mid#1#3#1\}}
\newcommand{\SetCompX}[3][]{\left\{#1#2#1\middle\vert#1#3#1\right\}}

\newcommand{\twoup}[1]{\ensuremath{2^{#1}}} 

\newcommand{\dom}[1]{\ensuremath{\mathrm{dom}(#1)}}

\renewcommand{\max}[1]{\ensuremath{\mathrm{max}(#1)}}

\newcommand{\Beh}[1]{\ensuremath{\mathcal{B}(#1)}}

\newcommand{\Spec}{\ensuremath{\psi}}
\newcommand{\Specr}{\ensuremath{\Spec_\mathrm{reach}}}
\newcommand{\Specs}{\ensuremath{\Spec_\mathrm{safe}}}

 \newcommand{\Behaclset}{\ensuremath{\mathcal{B}(\Saset^{cl})}}
 \newcommand{\Behtclset}{\ensuremath{\mathcal{B}(\Stset^{cl})}}

\newcommand{\hyint}[1]{\ensuremath{\llbracket #1 \,\rrbracket}}
\newcommand{\tn}[1]{{#1}}
\newcommand{\nindex}[1]{{#1}}
\newcommand{\n}[1]{{\eta_{#1}}}
\newcommand{\ta}[1]{{\tau_{#1}}}
\newcommand{\tindex}[1]{{#1}}
\newcommand{\frr}[1]{\preccurlyeq_{#1}}
\newcommand{\layer}{\ensuremath{{l\in[1;L]}}\xspace}
\newcommand{\llayer}{\ensuremath{{l'\in[1;L]}}\xspace}
\newcommand{\controllers}{\ensuremath{{p\in[1;P]}}\xspace}
\newcommand{\li}{\ensuremath{l_p}\xspace}

\newcommand{\etamin}{\ensuremath{\underline{\eta}}}
\newcommand{\taumin}{\ensuremath{\underline{\tau}}}

\newcommand{\St}[1]{\ensuremath{\overrightarrow{S}_{\tindex{#1}}}}
\newcommand{\Ft}[1]{\ensuremath{\overrightarrow{F}_{\tindex{#1}}}}

\newcommand{\Sa}[1]{\ensuremath{\widehat{S}_{\tn{#1}}}}
\newcommand{\Aa}[1]{\ensuremath{\widehat{A}^L_{\tn{#1}}}}
\newcommand{\Xa}[1]{\ensuremath{\widehat{X}_{\nindex{#1}}}}
\newcommand{\Oa}[1]{\ensuremath{\widehat{O}_{\nindex{#1}}}}

\newcommand{\Ua}{\ensuremath{\widehat{U}}}
\newcommand{\Fa}[1]{\ensuremath{\widehat{F}_{\tn{#1}}}}
\newcommand{\Ja}[1]{\ensuremath{\widehat{F}_{\tn{#1}}^L}}

\newcommand{\Qa}[1]{\ensuremath{\widehat{Q}_{\nindex{#1}}}}
\newcommand{\Qai}[1]{\ensuremath{\widehat{Q}^{-1}_{\nindex{#1}}}}

\newcommand{\Ra}[1]{\ensuremath{\widehat{R}_{#1}}}

\newcommand{\xa}{\ensuremath{\widehat{x}}}
\newcommand{\ya}{\ensuremath{\widehat{y}}}
\newcommand{\xia}{\ensuremath{\widehat{\xi}}}

\newcommand{\ua}{\ensuremath{\widehat{u}}}

\newcommand{\Ci}[1]{\ensuremath{C^{#1}}}
\newcommand{\Uci}[1]{\ensuremath{D^{#1}}}
\newcommand{\Gci}[1]{\ensuremath{G^{#1}}}

\newcommand{\C}{\ensuremath{C}}

\newcommand{\Uc}{\ensuremath{D}}

\newcommand{\Gc}{\ensuremath{G}}

\newcommand{\Saclall}{\ensuremath{\Saset^{cl}}}

\newcommand{\Stclall}{\ensuremath{\Stset^{cl}}}

\newcommand{\Xaall}{\ensuremath{\widehat{\textbf{X}}}}

\newcommand{\WIN}{\ensuremath{\mathcal{C}}}

\newcommand{\Aux}{\ensuremath{\Upsilon}}
\newcommand{\Cset}{\ensuremath{\mathbf{C}}}
\newcommand{\Ucset}{\ensuremath{\mathbf{D}}}

\newcommand{\Stset}{\ensuremath{\overrightarrow{\mathbf{S}}}}
\newcommand{\Saset}{\ensuremath{\widehat{\mathbf{S}}}}
\newcommand{\Aaset}{\ensuremath{\widehat{\mathbf{A}}}}

\newcommand{\Qset}{\ensuremath{\mathbf{Q}}}

\newcommand{\Ftall}{\ensuremath{\overrightarrow{\mathbf{F}}}}
\newcommand{\Faall}{\ensuremath{\widehat{\mathbf{F}}}}

\newcommand{\MLSAFE}{\mathsf{EagerSafe}\xspace}
\newcommand{\MLASAFE}{\mathsf{LazySafe}\xspace}
\newcommand{\SAFE}{\mathsf{Safe}\xspace}
\newcommand{\OTFAR}{\mathsf{ReachIteration}\xspace}
\newcommand{\MLREACH}{\mathsf{EagerReach}\xspace}
\newcommand{\REACH}{\mathsf{Reach}\xspace}

\newcommand{\COMPUTE}{\mathsf{ComputeTransitions}\xspace}
\newcommand{\EXPLORE}{\mathsf{ExpandAbstraction}\xspace}

\newcommand{\MLR}{\mathsf{EagerReach}\xspace}
\newcommand{\MLAR}{\mathsf{LazyReach}\xspace}
\newcommand{\OTFAS}{\mathsf{SafeIteration}\xspace}

\newcommand{\Tset}{\ensuremath{T}}

\newcommand{\Rset}{\ensuremath{R}}

\newcommand{\Tseta}[1]{\ensuremath{\widehat{\Tset}_{#1}}}

\newcommand{\Rseta}[1]{\ensuremath{\widehat{\Rset}_{#1}}}

\newcommand{\Ya}[1]{\ensuremath{\widehat{Y}_{\tn{#1}}}}

\def\goal{{\mathsf{G}}}
\def\obstacle{{\mathsf{O}}}

\usepackage{comment}
\usepackage{wrapfig}
\usepackage{setspace}
\usepackage{amsmath}
\usepackage{amssymb}
\usepackage{etoolbox}
\usepackage{color,colortbl}
\usepackage{mathtools}
\usepackage{enumerate}
\usepackage{algorithm, algpseudocode}
\usepackage{booktabs}
\usepackage[normalem]{ulem}
\usepackage{pgf}
\usepackage{tikz,pgfplots}
\usetikzlibrary{trees,decorations,arrows,arrows.meta,automata,shadows,positioning,plotmarks,backgrounds,shapes,shapes.misc}
\usetikzlibrary{calc,matrix,fit,petri,decorations.pathmorphing,patterns}
\usetikzlibrary{decorations.pathreplacing,decorations.markings,shapes.geometric,calc}
\usepackage{paralist}
\usepackage{stmaryrd}
\usepackage{xspace}
\usepackage{graphicx}
\usepackage{float}
\usepackage[hang, tight]{subfigure}
\usepackage[utf8]{inputenc} 
 \usepackage[english]{babel} 
  \usepackage{csquotes} 
\usepackage{tabularx}
\usepackage{wrapfig}

 \usepackage{array}

\usepackage{xifthen}
\usepackage{url}
\usepackage{relsize}

\newtheorem{assumption}{Assumption}
\makeatletter
\def\thanks#1{\protected@xdef\@thanks{\@thanks
        \protect\footnotetext{#1}}}
\makeatother
\begin{document}
\title{Lazy Abstraction-Based Controller Synthesis\thanks{
	This research was sponsored in part by the DFG project {389792660-TRR 248}
	and by the ERC Grant Agreement {610150} (ERC Synergy Grant ImPACT).
	Kyle Hsu was funded by a DAAD-RISE scholarship.
}
}

\author{Kyle Hsu\inst{1} \and Rupak Majumdar\inst{2} \and Kaushik Mallik\inst{2} \and Anne-Kathrin Schmuck\inst{2}}
\institute{
{University of Toronto, Canada\\
\email{\tt kyle.hsu@mail.utoronto.ca}
}
\and
{MPI-SWS, Kaiserslautern, Germany\\
\email{{\tt \{rupak,kmallik,akschmuck\}@mpi-sws.org}}
}
}

\maketitle

\setlength{\abovedisplayskip}{3pt}
\setlength{\belowdisplayskip}{3pt}


\section{Introduction}\label{sec:Introduction}

Abstraction-based controller synthesis (ABCS) is a general procedure for automatic synthesis of
controllers for continuous-time nonlinear dynamical systems against temporal specifications.
ABCS works by first abstracting a time-sampled version of the continuous dynamics of the open-loop system by a symbolic finite state model. 
Then, it computes  a finite-state controller for the symbolic model
using algorithms from automata-theoretic reactive synthesis.
When the time-sampled system and the symbolic model satisfy a certain refinement relation, 
the abstract controller can be refined to a controller for the original continuous-time system
while guaranteeing its time-sampled behavior satisfies the temporal specification. 
Since its introduction about 15 years ago, much research has gone into better theoretical understanding
of the basic method and extensions \cite{TabuadaBook,nilsson2017augmented,ReissigWeberRungger_2017_FRR,mallik2018compositional,coogan2015efficient,GruberKA17}, into scalable tools \cite{Scots,mouelhi2013cosyma,DBLP:conf/hybrid/KhaledZ19,li2018rocs}, and 
into demonstrating its applicability to nontrivial control problems \cite{ames2015first,DBLP:journals/tcst/NilssonHBCAGOPT16,saoud2018contract,borri2013decentralized}.

In its most common form, the abstraction of the continuous-time dynamical system is computed by fixing a parameter $\tau$ 
for the sampling time and a parameter $\eta$ for the state 
space,
and then representing the abstract state space as a set of hypercubes, each of diameter $\eta$. 
The hypercubes partition the continuous concrete state space. 
The abstract transition relation adds a transition between two hypercubes if there 
exists some state in the first hypercube and some control input that can reach some state of the second by following the original dynamics for time $\tau$. 
The transition relation is nondeterministic 
due to (a) the possibility of having continuous transitions starting at two different points in one hypercube but ending in different hypercubes, and
(b) the presence of external disturbances causing a deviation of the system trajectories from their nominal paths.
When restricted to a compact region of interest, the resulting finite-state abstract system describes a two-player game between controller 
and disturbance, and reactive synthesis techniques are used to algorithmically compute a controller (or show that no such controller exists for the given abstraction) for members of a broad class of temporal specifications against the disturbance.
One can show that the abstract transition system is in a \emph{feedback refinement relation} (FRR) with the original dynamics \cite{ReissigWeberRungger_2017_FRR}.
This ensures that when the abstract controller is applied to the original system, the time-sampled behaviors satisfy the temporal specifications.

The success of ABCS depends on the choice of $\eta$ and $\tau$. 
Increasing $\eta$ (and $\tau$) results in a smaller state space and symbolic model, but more nondeterminism.
Thus, there is a tradeoff between computational tractability and successful synthesis.
We have recently shown that one can explore this space of tradeoffs by maintaining multiple abstraction layers 
of varying granularity (i.e., abstract models constructed from progressively larger $\eta$ and $\tau$) \cite{hsu2018multi}. 
The multi-layered synthesis approach tries to find a controller for the coarsest abstraction whenever feasible, 
but adaptively considers finer abstractions when necessary. 
However, the bottleneck of our approach \cite{hsu2018multi} is that the abstract transition system of every layer 
needs to be fully computed before synthesis begins.
This is expensive and wasteful.
The cost of abstraction grows as $O((\frac{1}{\eta})^n)$, where $n$ is the dimension, and much of the abstract state space
may simply be irrelevant to ABCS, for example, if a controller was already found at a coarser level for a sub-region.

We apply the paradigm of \emph{lazy abstraction} \cite{henzinger2002lazy} to multi-layered synthesis for safety specifications in \cite{hsu2018lazy}.
Lazy abstraction is a technique to systematically and efficiently explore large state spaces through abstraction and refinement, 
and is the basis for successful model checkers for software, hardware, and timed systems 
\cite{beyer2007software,beyer2011cpachecker,vizel2012lazy,HSW13}. 
Instead of computing all the abstract transitions for the entire system in each layer, 
the algorithm selectively chooses which portions to compute transitions for, 
avoiding doing so for portions that have been already solved by synthesis. 
This co-dependence of the two major computational components of ABCS is both conceptually appealing and results in significant performance benefits.

This paper gives a concise presentation of the underlying principles of lazy ABCS enabling synthesis w.r.t.\ 
safety and reachability specifications.
Notably, the extension from single-layered to multi-layered and lazy ABCS is somewhat nontrivial, for the following reasons.

\textbf{(I) Lack of FRR between Abstractions.} 
An efficient multi-layered controller synthesis algorithm uses coarse grid cells almost everywhere in the state space and only resorts 
to finer grid cells where the trajectory needs to be precise. 
While this idea is conceptually simple, the implementation is challenging as the computation of such a multi-resolution controller domain via 
established abstraction-refinement techniques (as in, e.g., \cite{AlfaroRoy_2010}), requires one to run the fixed-point 
algorithms of reactive synthesis over a common game graph representation connecting abstract states of different coarseness. 
However, to construct the latter, a simulation relation must exist between any two abstraction layers.
Unfortunately, this is not the case in our setting: each layer uses a different sampling time
and, while each layer is an abstraction (at a different time scale) of the original system, layers may not have any FRR between themselves. Therefore, we can only run iterations of fixed-points within a particular abstraction layer, but not for combinations of them.
 
We therefore introduce novel fixed-point algorithms for safety and reach-avoid specifications.
Our algorithms save and re-load the results of one fixed-point iteration to and from the lowest (finest) abstraction
layer, which we denote as layer 1.
This enables arbitrary switching of layers between any two sequential iterations while reusing work from other layers. 
We use this mechanism to design efficient switching protocols which ensure that synthesis is done mostly over 
coarse abstractions while fine layers are only used if needed. 

\textbf{(II) Forward Abstraction and Backward Synthesis.} 
One key principle of lazy abstraction is that the abstraction is computed in the direction of the search. 
However, in ABCS, the abstract transition relation can only be computed \emph{forward},
since it involves simulating the ODE of the dynamical system forward up to the sampling time.
While an ODE can also be solved backwards in time,  
backward computations of reachable sets using numerical methods
may lead to high numerical errors \cite[Remark~1]{mitchell2007comparing}.
Forward abstraction conflicts with symbolic reactive synthesis algorithms, which work \emph{backward} by iterating controllable
predecessor operators.\footnote{
	One can design an enumerative forward algorithm for controller synthesis, essentially as a backtracking
	search of an AND-OR tree \cite{cassez2007efficient}, but dynamical perturbations greatly increase the width of the tree. 
	Experimentally, this leads to poor performance in control examples. 
}
For reachability specifications, we solve this problem by keeping a set of \emph{frontier} states, and proving that in the backward controllable
predecessor computation, all transitions that need to be considered arise out of these frontier states.
Thus, we can construct the abstract transitions lazily by computing the finer abstract transitions only for the frontier. 

\textbf{(III) Proof of Soundness and Relative Completeness.} 
The proof of correctness for common lazy abstraction techniques uses the property that there is a simulation
relation between any two abstraction layers \cite{ClarkeGJLV03,henzinger2003counterexample}. 
As this property does not hold in our setting (see (I)), our proofs of soundness and completeness w.r.t.\ the finest layer
only use (a) FRRs between any abstraction layer and the concrete system to argue about the correctness of a controller
in a sub-space, and combines this with
(b) an argument about the structure of ranking functions, which are obtained from 
fixed point iterations and combine the individual controllers.

\smallskip
\noindent\textbf{Related Work.}
Our work is an extension and consolidation of several similar attempts at using
multiple abstractions of varying granularity in the context of controller synthesis, including our own prior work \cite{hsu2018multi,hsu2018lazy}. 
Similar ideas were explored in the context of \emph{linear} dynamical systems \cite{aydin2012language,girard2006towards}, which enabled 
the use of polytopic approximations.
For \emph{unperturbed} systems \cite{aydin2012language,rungger2012fly,grune1997adaptive,pola2012integrated,CameraGirardGoessler_reach_2011,Girard2016_InterSampling},
one can implement an efficient forward search-based synthesis technique, thus easily enabling lazy abstraction (see (II) above).
For nonlinear systems satisfying a stability property, \cite{Girard2016_InterSampling,CameraGirardGoessler_safety_2011,CameraGirardGoessler_reach_2011} show
a multi-resolution algorithm. It is implemented in the tool CoSyMA \cite{mouelhi2013cosyma}. 
For \emph{perturbed nonlinear systems}, \cite{fribourg2013constructing,fribourg2014finite} show a successive decomposition technique based on a linear
approximation of given polynomial dynamics.
Recently, Nilsson et al.~\cite{nilsson2017augmented,bulancea2018nonuniform} presented an abstraction-refinement technique for perturbed nonlinear systems 
which shares with our approach the idea of using particular candidate states for \emph{local} refinement. 
However, the approaches differ by the way abstractions are constructed.
The approach in \cite{nilsson2017augmented} identifies all adjacent cells which are reachable using a particular input, 
splitting these cells for finer abstraction computation. 
This is analogous to established abstraction-refinement techniques for solving two-player games \cite{AlfaroRoy_2010}. 
On the other hand, our method computes reachable sets for particular sampling times that 
vary across layers, resulting in a more delicate abstraction-refinement loop. 

\begin{figure*}[t!]
	\centering
		\input{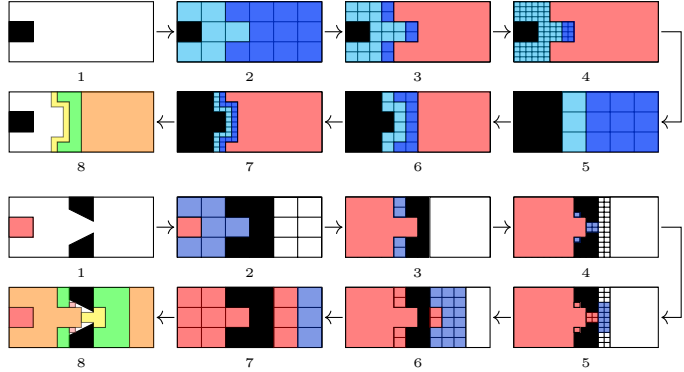} 
	\caption{An illustration of the lazy ABCS algorithms for safety (top) and reach-avoid (bottom) specifications.
	In both scenarios, the solid black regions are the unsafe states which need to be avoided. In the reach-avoid problem, the system has to additionally reach the target ($T$) red square at the left of Pic.~1.
	Both figures show the sequence of synthesis stages across three abstraction layers: $l=1$ (Pics.~4,~7), $l=2$ (Pics.~3,~6), and $l=3$ (Pics.~2,~5) for safety; and $l=1$ (Pics.~4,~5), $l=2$ (Pics.~3,~6), and $l=3$ (Pics.~2,~7) for reach-avoid. 
	Both Pics.~8 indicate the domains of the resulting controllers with different granularity: $l=1$ (yellow), $l=2$ (green), and $l=3$ (orange).
	The red regions represent winning states, and the blue regions represent states determined as winning in the present synthesis stage. Cyan regions represent ``potentially losing'' states in the safety synthesis.
	We set the parameter $m=2$ for reach-avoid synthesis.
	The gridded regions in different layers represent the states where the transitions have been computed; large ungridded space in $l=2$ and $l=1$ signifies the computational savings of the lazy abstraction approach.
	}
	\label{fig:informal example}
	\vspace{-0.5cm}
\end{figure*}


\smallskip
\noindent\textbf{Informal Overview.}
We illustrate our approach via solving a safety control problem and a reach-avoid control problem depicted in \REFfig{fig:informal example}.
In reactive synthesis, these problems are solved using a maximal and minimal fixed point computation, respectively \cite{MPS95}.
Thus, for safety, one starts with the maximal set of safe states and iteratively shrinks the latter until the remaining set, 
called the winning state set, does not change.
That is, for all states in the winning state set, there is a control action which ensures that the system remains within this set for one step. 
For reachability, one starts with the set of target states as the winning ones and iteratively enlarges this set by adding all states 
which allow the system to surely reach the current winning state set, until no more states can be added. 
These differences in the underlying characteristics of the fixed points require different switching protocols when multiple abstraction layers are used.

The purpose of \REFfig{fig:informal example}
is only to convey the basic idea of our algorithms in a visual and lucid way, 
without paying attention to the details of the underlying dynamics of the system.
In our example, we use three layers of abstraction $S_1$, $S_2$ and $S_3$ with the parameters $(\eta,\tau)$, $(2\eta,2\tau)$ and $(4\eta,4\tau)$.
We refer to the steps of \REFfig{fig:informal example} as Pic.~\#.

For the \emph{safety control problem} (the top figure in \REFfig{fig:informal example}), we assume that a set of unsafe states are given (the black box in the left of Pic.~1).
These need to be avoided by the system.
For lazy ABCS, we first fully compute the abstract transition relation of the coarsest abstraction $S_3$, 
and find the states from where the unsafe states can be avoided for at least one time step of length $4\tau$ (blue region in Pic.~2).
Normally for a single-layered algorithm, the complement of the blue states would immediately be discarded as losing states.
However, in the multi-layered approach, we treat these states as \emph{potentially losing} (cyan regions), and proceed to $S_2$ (Pic.~3) 
to determine if some of these potentially losing states can avoid the unsafe states with the help of a more fine-grained controller.

However, we cannot perform any safety analysis on $S_2$ yet as the abstract transitions of $S_2$ have not been computed.
Instead of computing all of them, as in a non-lazy approach, we only locally explore the outgoing transitions of the potentially losing states in $S_2$.
Then, we compute the subset of the potentially losing states in $S_2$ that can avoid the unsafe states for at least one time step (of length $2\tau$ in this case).
These states are represented by the blue region in Pic.~3, which get saved from being discarded as losing states in this iteration.
Then we move to $S_1$ with the rest of the potentially losing states and continue similarly.
The remaining potentially losing states at the end of the computation in $S_1$ are surely losing---relative to the finest abstraction $S_1$---and are permanently discarded.
This concludes one ``round'' of exploration.

We restart the process from $S_3$.
This time, the goal is to avoid reaching the unsafe states for at least two time steps of available lengths.
This is effectively done by inflating the unsafe region with the discarded states from previous stages (black regions in Pics.~5,~6, and 7).
The procedure stops when the combined winning regions across all layers do not change for two successive iterations.

In the end, the multi-layered safety controller is obtained as a collection of the safety controllers 
synthesized in different abstraction layers in the last round of fixed-point computations.
The resulting safety controller domain is depicted in Pic.~8.

Now consider the \emph{reach-avoid control problem} in \REFfig{fig:informal example} (bottom). 
The target set is shown in red, and the states to be avoided in black.
We start by computing the abstract transition system completely for the coarsest layer and solve the reachability
fixed point at this layer until convergence using under-approximations of the target and of the safe states.
The winning region is marked in blue (Pic.~2); note that the approximation of the bad states ``cuts off'' the possibility
to reach the winning states from the states on the right.
We store the representation of this winning region in the finest layer as the set $\Aux_1$.

%

Intuitively, we run the reachability fixed point until convergence to enlarge the winning state set as much as possible using large cells. 
This is in contrast to the previous algorithm for safety control in which we performed just one iteration at each level.
For safety, each iteration of the fixed-point shrinks the winning state set. 
Hence, running $\SAFE$ until convergence would only keep those coarse cells which form an invariant set by themselves. 
Running one iteration of $\SAFE$ at a time instead has the effect that clusters of finer cells which can be controlled to be safe by 
a suitable controller in the corresponding layer
are considered safe in the coarser layers in future iterations.
This allows the use of coarser control actions in larger parts of the state space (see Fig.~1 in~\cite{hsu2018lazy} for an illustrative example of this phenomenon).

To further extend the winning state set $\Aux_1$ for reach-avoid control, we proceed to the next finer layer
$l=2$ with the new target region (red) being the projection of $\Aux_1$ to $l=2$.
As in safety control, all the safe states in the complement of $\Aux_l$ are potentially within the winning state set. 
The abstract transitions at layer $l=2$ have not been computed at this point.
We only compute the abstract transitions for the \emph{frontier} states: these are all the cells that
might contain layer $2$ cells that can reach the current winning region within $m$ steps (for some parameter $m$ chosen in the implementation).
The frontier is indicated for layer $2$ by the small gridded part in Pic.~3.

We continue the backward reachability algorithm on this partially computed transition system by running the fixed-point for $m$ steps.
The projection of the resulting states to the finest layer is added to $\Aux_1$.  
In our example (Pic.~3), we reach a fixed-point just after $1$ iteration implying that no more layer $2$ (or layer $3$) 
cells can be added to the winning region. 

We now move to layer $1$, compute a new frontier (the gridded part in Pic.~4), and run the reachability fixed point on $\Aux_1$ for $m$ steps.
We add the resulting winning states to $\Aux_1$ (the blue region in Pic.~4). 
At this point, we could keep exploring and synthesizing in layer $1$, 
but in the interest of efficiency we want to give the coarser layers a chance to progress. 
This is the reason to only compute $m$ steps of the reachability fixed point in any one iteration.
Unfortunately, for our example, the attempt to go coarser fails as no new layer $2$ cells can be added yet (see Pic.~3). 
We therefore fall back to layer $1$ and make progress for $m$ more steps (Pic.~5). 
At this point, the attempt to go coarser is successful (Pic.~6) as the right side of the small passage was reached.

We continue this movement across layers until synthesis converges in the finest layer. 
In Pic.~8, the orange, green and yellow colored regions are the controller domains 
obtained using $l=3$, $l=2$ and $l=1$, respectively. 
Observe that we avoid computing transitions for a significant portion of layers $1$ and $2$ (the ungridded space in Pics.~5, 6, respectively).

\section{Control Systems and Multi-Layered ABCS}\label{sec:prelim}

We recall the theory of feedback refinement relations (FRR) \cite{ReissigWeberRungger_2017_FRR} and multi-layered ABCS \cite{hsu2018multi}.

\smallskip
\noindent\textbf{Notation.}
We use the symbols $\N$, $\real{}$, $\real{}_{>0}$, $\mathbb{Z}$, and $\mathbb{Z}_{>0}$ 
to denote the sets of natural numbers, reals, positive reals, integers, and positive integers, respectively. 
Given $a,b\in\real{}$ with $a\leq b$, we write $[a,b]$ for the closed interval $\set{x\in\real{}\mid a\leq x \leq b}$ and
write $[a;b]=[a,b]\cap \Z$ as its discrete counterpart. 
Given a vector $a\in\real{n}$, we denote by $a_{i}$ its $i$-th element, for $i\in [1;n]$.
We write $\hyint{a,b}$ for the closed hyper-interval $\real{n}\cap([a_1,b_1]\times\hdots\times[a_n,b_n])$.
We define the relations $<,\leq,\geq,>$ on vectors in $\real{n}$ component-wise.
For a set $W$, we write $W^*$ and $W^\omega$ for the sets of finite
and infinite sequences over $W$, respectively.
We define  $W^\infty = W^* \cup W^\omega$.
We define $\dom{w} = \Set{0,\ldots, |w|-1}$ if $w\in W^*$, and $\dom{w} = \N$ if $w\in W^\omega$. 
For $k\in \dom{w}$ we write $w(k)$ for the $k$-th symbol of $w$.

\subsection{Abstraction-Based Controller Synthesis}\label{sec:Prelim_FRR}

\smallskip
\noindent\textbf{Systems.} 
A \emph{system} $S=(X,U,F)$ consists of a state space $X$, an input space $U$, and a transition function $F:X\times U\fun 2^X$. 
A system $S$ is \emph{finite} if $X$ and $U$ are finite. 
A trajectory $\xi\in X^\infty$ is a maximal sequence of states compatible with $F$:
for all $1\leq k < |\xi|$ there exists $u\in U$ s.t.\ $\xi(k)\in F(\xi(k-1),u)$ and 
if $|\xi| < \infty$ then $F(\xi(|\xi|),u)= \emptyset$ for all $u\in U$.
For $D\subseteq X$, a $D$-trajectory is a trajectory $\xi$ with $\xi(0)\in D$.
The \emph{behavior} $\Beh{S, D}$ of a system $S=(X,U,F)$ w.r.t.\ $D\subseteq X$ 
consists of all $D$-trajectories; when $D=X$, we simply write $\Beh{S}$.

\smallskip
\noindent\textbf{Controllers and Closed Loop Systems.}
A \emph{controller} $\C=(\Uc,U,\Gc)$ for a system $S=(X,U,F)$
consists of a controller \emph{domain} $\Uc\subseteq X$, a space of inputs $U$, and a control map 
$\Gc:\Uc\fun 2^{U}\setminus\set{\emptyset}$ 
mapping states in its domain to non-empty sets of control inputs.
The \emph{closed loop system} formed by interconnecting
$S$ and $\C$ in \emph{feedback} is defined by the system 
 $S^{cl}=(X,U, F^{cl})$ with $F^{cl}:X\times U \fun 2^{X}$ s.t.\ 
 $x'\in F^{cl}(x,u)$ iff $x\in \Uc$ and $u\in\Gc(x)$ and $x'\in F(x,u)$, or $x\notin \Uc$ and $x'\in F(x,u)$.

%

\smallskip
\noindent\textbf{Control Problem.}
We consider specifications given as $\omega$-regular languages 
whose atomic predicates are interpreted as sets of states. 
Given a specification $\psi$, a system $S$, and an interpretation of the predicates as sets
of states of $S$, we write $\semantics{\psi}_S\subseteq \Beh{S}$ 
for the set of behaviors of $S$ satisfying $\psi$.
The pair $\tuple{S,\psi}$ is called a \emph{control problem} on $S$ for $\psi$. 
A controller $\C = (\Uc, U, G)$ for $S$ solves $\tuple{S,\psi}$ if
$\Beh{S^{cl}, \Uc}\subseteq\semantics{\psi}_S$.
The set of all controllers solving $\tuple{S,\psi}$ is denoted by $\WIN(S, \psi)$.

\smallskip
\noindent\textbf{Feedback Refinement Relations.}
Let $S_i=(X_i,U_i,F_i)$, $i\in\Set{1,2}$ be two systems with
$U_2\subseteq U_1$. 
A \emph{feedback refinement relation} (FRR) from $S_1$ to $S_2$ 
is a relation $Q\subseteq X_1\times X_2$ s.t.\ 
for all $x_1\in X_1$ there is some $x_2\in X_2$ such that $Q(x_1,x_2)$ and
for all $(x_1,x_2)\in Q$, we have
\begin{inparaenum}[(i)]
 \item $U_{S_2}(x_2)\subseteq U_{S_1}(x_1)$, and 
 \item $u\in U_{S_2}(x_2) \Rightarrow Q(F_1(x_1,u))\subseteq F_2(x_2,u)$
\end{inparaenum}
where $U_{S_i}(x):=\SetComp{u\in U_i}{F_i(x,u)\neq \emptyset}$.
%
We write $S_1\frr{Q} S_2$ if $Q$ is an FRR from $S_1$ to $S_2$.

\smallskip
\noindent\textbf{Abstraction-Based Controller Synthesis (ABCS).}
Consider two systems $S_1$ and $S_2$, with $S_1\frr{Q} S_2$. 
Let $\C=(\Uc,U_2,\Gc)$ be a controller for $S_2$. 
Then, as shown in \cite{ReissigWeberRungger_2017_FRR}, $\C$ 
can be refined into a controller for $S_1$, defined by 
$\C\circ Q=(\widetilde{\Uc},U_1,\widetilde{\Gc})$ 
with $\widetilde{\Uc}= Q^{-1}(\Uc)$, and
$\widetilde{\Gc}(x_1)=\SetComp{u\in U_1}{\ExQ{x_2\in Q(x_1)}{u\in\Gc(x_2)}}$ for all $x_1\in\widetilde{\Uc}$. 
This implies soundness of ABCS. 

\begin{proposition}[\cite{ReissigWeberRungger_2017_FRR}, Def. VI.2, Thm. VI.3]
\label{prop:paimpliespt}
Let $S_1\frr{Q} S_2$ and $\C\in\WIN(S_2,\psi)$ for a specification $\psi$.
If for all $\xi_1\in \Beh{S_1}$ and $\xi_2\in\Beh{S_2}$ with
$\dom{\xi_1}=\dom{\xi_2}$ and $(\xi_1(k), \xi_2(k))\in Q$
for all $k\in \dom{\xi_1}$ holds that  
$\xi_2\in\semantics{\psi}_{S_2}\Rightarrow\xi_1\in\semantics{\psi}_{S_1}$, 
then $\C\circ Q \in \WIN(S_1, \psi)$.
\end{proposition}


\subsection{ABCS for Continuous Control Systems}\label{sec:prelim_ABCS}

We now recall how ABCS can be applied to continuous-time systems by delineating the abstraction procedure \cite{ReissigWeberRungger_2017_FRR}. 

\smallskip
\noindent\textbf{Continuous-Time Control Systems.}\
A \emph{control system} $\Sigma = (X, U, W, f)$
consists of a state space $X= \real{n}$, a non-empty input space $U\subseteq\real{m}$, 
a compact disturbance set $W\subset \real{n}$ with $0\in W$, and 
a function $f:X\times U\rightarrow X$ s.t. $f(\cdot,u)$ is locally Lipschitz for all $u\in U$. 
Given an initial state $x_0\in X$, a positive parameter $\tau>0$, and a constant input trajectory $\mu_u:[0,\tau]\rightarrow U$
which maps every $t\in [0,\tau]$ to the same $u\in U$, 
a \emph{trajectory} of $\Sigma$ 
on $[0,\tau]$ is an absolutely continuous function $\xi:[0,\tau]\rightarrow X$  s.t. $\xi(0) = x_0$ and
$\xi(\cdot)$ fulfills the following differential inclusion for almost every $t\in[0,\tau]$:
\begin{equation}\label{equ:def_f}
 \dot{\xi}\in f(\xi(t),\mu_u(t))+W = f(\xi(t),u) + W. 
\end{equation} 
We collect all such solutions in the set $\ON{Sol}_f(x_0,\tau,u)$.

\smallskip
\noindent\textbf{Time-Sampled System.}\
Given a time sampling parameter $\tau>0$, we define the \emph{time-sampled system} $\St{}(\Sigma,\tau)=(X,U,\Ft{})$ associated with $\Sigma$, 
where $X$, $U$ are as in $\Sigma$, and the transition function $\Ft{}:X\times U\fun 2^X$ is defined as follows.
For  all $x\in X$ and $u \in U$, we have $x'\in \Ft{}(x,u)$ iff there exists a solution $\xi\in\ON{Sol}_f(x,\tau,u)$ s.t.\ $\xi(\tau)=x'$.

\smallskip
\noindent\textbf{Covers.}\ 
A \emph{cover} $\widehat{X}$ of the state space $X$ is a set of
non-empty, closed hyper-intervals $\hyint{a,b}$ with $a,b\in (\real{}\cup\Set{\pm\infty})^n$ called \emph{cells},
such that every $x\in X$ belongs to some cell in $\widehat{X}$. 
%
%
Given a grid parameter $\eta \in\real{}_{>0}^n$, we say that a point $c\in Y$ is \emph{$\eta$-grid-aligned} if there is $k\in\Z^n$ s.t.\ for each $i\in \set{1,\ldots,n}$,
$c_i = \alpha_i + k_i\eta_i - \frac{\eta_i}{2}$.
Further, a cell $\hyint{a,b}$ is \emph{$\eta$-grid-aligned} if there is a $\eta$-grid-aligned point $c$ s.t.\ $a = c - \frac{\eta}{2}$ and
$b = c + \frac{\eta}{2}$;
such cells define sets of diameter $\eta$ whose center-points are $\eta$-grid-aligned. 



\smallskip
\noindent\textbf{Abstract Systems.}\ 
An \emph{abstract system} $\Sa{}(\Sigma,\tau,\eta)=(\Xa{},\Ua{},\Fa{})$ for a control system $\Sigma$,
a time sampling parameter $\tau > 0$, and a grid parameter $\eta \in \real{n}_{>0}$
consists of an abstract state space $\Xa{}$, a finite abstract input space $\Ua\subseteq U$, and an abstract transition function 
$\Fa{}:\Xa{}\times \Ua{}\rightarrow 2^{\Xa{}}$. To ensure that $\Sa{}$ is finite, we consider a compact \emph{region of interest} $Y=\hyint{\alpha, \beta}\subseteq X$ with $\alpha,\beta\in\real{n}$ s.t.\ $\beta - \alpha$ is an integer multiple of $\eta$. 
Then we define $\Xa{}=\widehat{Y}\cup\widehat{X}'$ s.t.\ $\widehat{Y}$ is the \emph{finite} set of \emph{$\eta$-grid-aligned} cells covering $Y$ and $\widehat{X}'$ is a finite set of large unbounded cells covering the (unbounded) region $X\setminus Y$. 
We define $\Fa{}$ based on the dynamics of $\Sigma$ only within $Y$. That is, for all $\xa\in\Ya{}$, $\xa'\in\Xa{}$, and $u\in\Ua{}$ 
we require
 \begin{align}\label{eq:next state abs sys 0}
   \xa'\in \Fa{}(\xa,u) \ \mbox{ if } \ 	\exists\xi\in \cup_{x\in\xa}\ON{Sol}_f(x,\tau,u) \;.\; \xi(\tau) \in \xa'.
 \end{align}
For all states in  $\xa\in(\Xa{}\setminus\Ya{})$ we have that $\Fa{}(\xa,u)=\emptyset$ for all $u\in\Ua$.
We extend $\Fa{}$ to sets of abstract states $\Upsilon\subseteq \Xa{}$ by defining $\Fa{}(\Upsilon,u) := \bigcup_{\xa\in \Upsilon} \Fa{}(\xa,u)$.

While $\Xa{}$ is not a partition of the state space $X$, notice that cells only overlap at the boundary and one can define 
a deterministic function that resolves the resulting non-determinism by consistently mapping such boundary states to a unique cell covering it. 
The composition of $\Xa{}$ with this function defines a partition.
To avoid notational clutter, we shall simply treat $\widehat{X}$ as a partition.


\smallskip
\noindent\textbf{Control Problem.}\
It was shown in \cite{ReissigWeberRungger_2017_FRR}, Thm. III.5 that the relation
$\Qa{}\subseteq X\times \Xa{}$, defined by all tuples $(x,\xa)\in\Qa{}$ for which $x\in\xa$,
is an FRR between $\St{}$ and $\Sa{}$, i.e.,
$\St{}\frr{\Qa{}}\Sa{}$.
Hence, we can apply ABCS as described in \REFsec{sec:Prelim_FRR} 
by computing a controller $C$ for $\Sa{}$ which can then be refined to a controller for $\St{}$ under the pre-conditions of Prop.~\ref{prop:paimpliespt}.

More concretely, we consider safety and reachability control problems for the 
continuous-time system $\Sigma$, which are defined by a set of \emph{static obstacles} $\obstacle \subset X$ which 
should be avoided and a set of \emph{goal states} $\goal\subseteq X$ which should be reached, respectively.
Additionally, when constructing $\Sa{}$, we used a compact region of interest $Y\subseteq X$ to ensure \emph{finiteness} of $\Sa{}$ allowing to apply tools from reactive synthesis \cite{MPS95} to compute $C$. 
This implies that $C$ is only valid within $Y$. 
We therefore interpret $Y$ as a \emph{global safety requirement} and synthesize a controller which keeps the system within 
$Y$ while implementing the specification. 
This interpretation leads to a safety and reach-avoid
control problem, w.r.t.\ a safe set $\Rset=Y\setminus \obstacle$ and target set $\Tset=\goal \cap\Rset$.
As $\Rset{}$ and $\Tset{}$ can be interpreted as predicates over the state space $X$ of $\St{}$, this directly defines the control problems $\tuple{\St{},\Specs}$ and $\tuple{\St{},\Specr}$ via
\begin{subequations}\label{equ:Csound}
 \begin{align}
 \semantics{\Specs}_{\St{}}&:=\SetCompX{\xi\in\Beh{\St{}}}{\AllQ{k\in\dom{\xi}}{\xi(k)\in\Rset{}}},~\text{and}\label{equ:Csound:safe}\\
 \semantics{\Specr}_{\St{}}&:=\SetCompX{\xi\in\Beh{\St{}}}{\ExQ*{k\in\dom{\xi}}{
 \begin{propConjA}
  \xi(k)\in\Tset{}\\
  \AllQ{k' \leq k}{\xi(k')\in \Rset{}}
 \end{propConjA}
 }}\label{equ:Csound:reach}
\end{align}
\end{subequations}
for safety and reach-avoid control, respectively.
Intuitively, a controller $C\in\WIN(\St{},\Spec)$ applied to $\Sigma$ is a sample-and-hold controller, which ensures that the specification holds on all closed-loop trajectories \emph{at sampling instances}.%
\footnote{This implicitly assumes that sampling times and grid sizes are such that no \enquote{holes} occur between consecutive cells visited by a trajectory. This can be formalized by assumptions on the growth rate of $f$ in \eqref{equ:def_f} which is beyond the scope of this paper.}

To compute $C\in\WIN(\St{},\Spec)$ via ABCS as described in \REFsec{sec:Prelim_FRR} we need to ensure that the pre-conditions of Prop.~\ref{prop:paimpliespt} hold. This is achieved by \emph{under-approximating} the safe and target sets by abstract state sets 
\begin{align}\label{equ:RsetaTseta}
 \Rseta{}=\SetComp{\xa\in\Xa{}}{\xa\subseteq \Rset{}},~\text{and}~
 \Tseta{}=\SetComp{\xa\in\Xa{}}{\xa\subseteq \Tset{}},
\end{align}
and defining $\semantics{\Specs}_{\Sa{}}$ and $\semantics{\Specr}_{\Sa{}}$ via \eqref{equ:Csound} by substituting $\St{}$ with $\Sa{}$, $R$ with $\Rseta{}$ and $T$ with $\Tseta{}$. With this, it immediately follows from Prop.~\ref{prop:paimpliespt} that $C\in\WIN\tuple{\Sa{},\Spec}$ can be refined to the controller $C\circ Q\in\WIN{\tuple{\St{},\Spec}}$.

%

\subsection{Multi-Layered ABCS}\label{sec:prelim_Multi}

We now recall how ABCS can be performed over multiple abstraction layers \cite{hsu2018multi}.
The goal of multi-layered ABCS is to construct an abstract controller $C$ which uses coarse grid cells in as much part of the state space as possible, and only resorts to finer grid cells where the control action needs to be precise. In particular, the domain of this abstract controller $C$ must not be smaller then the domain of any controller $C'$ constructed for the finest layer, i.e., $C$ must be relatively complete w.r.t.\ the finest layer. 
In addition, $C$ should be refinable into a controller implementing $\psi$ on $\Sigma$, as in classical ABCS (see Prop.~\ref{prop:paimpliespt}).

The computation of such a multi-resolution controller via established abstraction-refinement techniques (as in, e.g., \cite{AlfaroRoy_2010}), requires 
a common transition system  connecting states of different coarseness but with the same time step. 
To construct the latter, a FRR between any two abstraction layers must exist, which is not the case in our setting. 
We can therefore not compute a single multi-resolution controller $C$. 
We therefore synthesize a set $\Cset$ of single-layered controllers instead, each for a different coarseness and with a different domain, 
and refine each of those controllers separately, using the associated FRR. 
The resulting refined controller is a sample-and-hold controller which selects the current input value $u\in\Ua\subseteq U$ and 
the duration $\tau_l$ for which this input should be applied to $\Sigma$. This construction is formalized in the remainder of this section.

\smallskip
\noindent\textbf{Multi-Layered Systems.}\
Given a grid parameter $\etamin$, a time sampling parameter $\taumin$, 
and $L \in \mathbb{Z}_{>0}$, 
define
%
$\n{l} = 2^{l-1}\etamin$ and $\ta{l} = 2^{l-1}\taumin$.
For a control system $\Sigma$ and a subset $Y\subseteq X$ with $Y=\hyint{\alpha,\beta}$, 
s.t.\ $\beta-\alpha=k\n{L}$ for some $k\in \Z^n$, $\widehat{Y}_l$ is the $\n{l}$-grid-aligned cover of $Y$.
%
This induces
a sequence of time-sampled systems 
\begin{align}\label{equ:StsetSaset}
 \Stset:=\set{\St{l}(\Sigma,\ta{l})}_\layer\quad \text{and}\quad \Saset:=\set{\Sa{l}(\Sigma,\ta{l},\n{l})}_\layer,
\end{align}
respectively, where $\St{l}:=(X,U,\Ft{l})$ and $\Sa{l}:=(\Xa{l},\Ua{},\Fa{l})$. If $\Sigma$, $\tau$, and $\eta$ are
clear from the context, we omit them in $\St{l}$ and $\Sa{l}$.
 
Our multi-layered synthesis algorithm relies on the fact that the sequence $\Saset{}$ of abstract transition systems is monotone, formalized by the following assumption.%
\begin{assumption}\label{assumption:monotonic approximation}
Let $\Sa{l}$ and $\Sa{m}$ be two abstract systems with $m,\layer$, $l<m$. Then\footnote{We write $\Upsilon_{l}\subseteq \Upsilon_{m}$ with $\Upsilon_{l}\subseteq \Xa{l}$, $\Upsilon_{m}\subseteq \Xa{m}$ as short for $\bigcup_{\xa\in \Upsilon_l}\xa\subseteq \bigcup_{\xa\in \Upsilon_m}\xa$.} $\Fa{l}(\Upsilon_l) \subseteq \Fa{m}(\Upsilon_m)$ if $\Upsilon_{l}\subseteq \Upsilon_{m}$.
\end{assumption} 
As the exact computation of $\cup_{x\in\xa}\ON{Sol}_f(x,\tau,u)$ in \eqref{eq:next state abs sys 0} is expensive (if not impossible), a numerical over-approximation is usually computed. \REFass{assumption:monotonic approximation} states that the approximation must be monotone in the granularity of the discretization. This is fulfilled by many numerical methods e.g.\ the ones based on decomposition functions for mixed-monotone systems \cite{coogan2015efficient} or on growth bounds \cite{Scots}; our implementation uses the latter one.

%
%

\smallskip
\noindent\textbf{Induced Relations.}\
It trivially follows from our construction that, for all $\layer$, 
we have $\St{l}\frr{\Qa{l}}\Sa{l}$, where $\Qa{l} \subseteq X\times\Xa{l}$ is the FRR induced by $\Xa{l}$.
The set of relations $\{\Qa{l}\}_\layer$ induces transformers $\Ra{ll'}\subseteq\Xa{l}\times\Xa{l'}$ for $l,\llayer$ between abstract states of different layers such that
\begin{align}\label{equ:Ra}
 \propAequ{\xa\in\Ra{ll'}(\xa')}{\xa\in \Qa{l}(\Qai{l'}(\xa')).}
\end{align}
However, the relation $\Ra{ll'}$ is generally \emph{not} a FRR between the layers due to different time sampling parameters used in different layers (see \cite{hsu2018multi}, Rem. 1). 
This means that $\Sa{l+1}$ cannot be directly constructed from $\Sa{l}$, unlike in usual abstraction refinement algorithms \cite{ClarkeGJLV03,henzinger2002lazy,AlfaroRoy_2010}.

\smallskip
\noindent\textbf{Multi-Layered Controllers.}\
Given a multi-layered abstract system $\Saset$ and some $P\in \mathbb{N}$, a multi-layered controller is a set  $\Cset=\set{\Ci{p}}_{\controllers}$ with $\Ci{p}=(\Uci{p}, \Ua,\Gci{p})$ being a single-layer controller with $\Gci{p}:\Uci{p}\fun 2^{\Ua}$. Then $\Cset$
is a controller for $\Saset$ if for all $\controllers$ there exists a unique $\li\in[1;L]$ s.t.\ $\Ci{p}$ is a controller for $\Sa{\li}$, i.e., $\Uci{p}\subseteq \Xa{\li}$. 
%
The number $P$ may not be related to $L$; we allow multiple controllers for the same layer and no controller for some layers.

The \emph{quantizer induced by $\Cset$} is a
map $\Qset:X\fun 2^{\Xaall}$ with $\Xaall=\bigcup_\layer \Xa{l}$ s.t.\ for all $x\in X$ it holds that $\xa\in\Qset(x)$ iff 
 there exists $p\in[1;P]$ s.t.\  $\xa\in\Qa{l_p}(x)\cap\Uci{p}$
 and no $p'\in[1;P]$ s.t.\ $l_{p'}>l_p$ and $\Qa{l_{p'}}(x)\cap\Uci{p'}\neq\emptyset$.
%
%
In words, $\Qset$ maps states $x\in X$ to the \emph{coarsest} abstract state $\xa$ that is both related to $x$ and is in the domain $\Uci{p}$ of some $\Ci{p}\in\Cset$.
We define $\Ucset=\SetComp{\xa\in\Xaall}{\ExQ{x\in X}{\xa\in\Qset(x)}}$ as the effective domain of $\Cset$ and $\Uc=\SetComp{x\in X}{\Qset(x)\neq\emptyset}$ as its projection to $X$.

\smallskip
\noindent\textbf{Multi-Layered Closed Loops.}\
The abstract multi-layered \emph{closed loop system} formed by interconnecting
$\Saset$ and $\Cset$ in \emph{feedback} 
is defined by the system 
$\Saset^{cl}=(\Xaall, \Ua{}, \Faall^{cl})$ with $\Faall^{cl}:\Xaall\times\Ua\fun2^{\Xaall}$ s.t.\ 
$\xa'\in\Faall^{cl}(\xa,\ua)$ iff
(i) there exists $p\in[1;P]$ s.t.\ $\xa\in\Ucset\cap\Uci{p}$, $\ua\in\Gci{p}(\xa)$ and there exists $\xa''\in\Fa{\li}(\xa,\ua)$ s.t.\ either $\xa'\in\Qset(\Qai{\li}(\xa''))$, or $\xa'=\xa''$ and $\Qai{\li}(\xa'')\not\subseteq \Uc$, or 
(ii) $\xa\in\Xa{l}$, $\Qai{l}(\xa)\not\subseteq \Uc$ and $\xa'\in\Fa{l}(\xa,\ua)$.
%
%
This results in the time-sampled closed loop system
$\Stset^{cl}=(X,U, \Ftall^{cl})$ with $\Ftall^{cl}:X\times U\fun 2^X$ s.t.\ 
$x'\in\Ftall^{cl}(x,u)$ iff (i) $x\in \Uc$ and there exists $p\in[1;P]$ and $\xa\in \Qset(x)\cap\Uci{p}$ s.t.\ $u\in\Gci{p}(\xa)$ and $x'\in\Ft{\li}(x,u)$, or (ii) $x\notin \Uc$ and $x'\in\Ft{l}(x,u)$ for some $\layer$.
%

\smallskip
\noindent\textbf{Multi-Layered Behaviors.}\
Slightly abusing notation, we define the behaviors $\Beh{\Stset}$ and $\Beh{\Saset}$ via the construction for systems $S$ in \REFsec{sec:Prelim_FRR} by interpreting the sequences $\Stset$ and $\Saset$ as systems 
$\Saset=(\Xaall,\Ua,\Faall)$ and
 $\Stset=(X,U,\Ftall)$, s.t.\,
\begin{align}
 \textstyle\Faall(\xa,\ua)=\bigcup_{\layer}\Ra{ll'}(\Fa{l'}(\xa,\ua)),~\text{and}\quad\textstyle\Ftall(x,u)=\bigcup_{\layer}\Ft{l}(x,u),
\end{align}
where $\xa$ is in $\Xa{l'}$.
Intuitively, the resulting behavior $\Beh{\Saset}$ contains trajectories with non-uniform state size; in every time step the system can switch to a different layer using the available transition functions $\Fa{l}$. For $\Beh{\Stset}$ this results in trajectories with non-uniform sampling time; in every time step a transition of any duration $\ta{l}$ can be chosen, which corresponds to some $\Ft{l}$.
For the closed loops $\Stset^{cl}$ and $\Saset^{cl}$ those behaviors are restricted to follow the switching pattern induced by $\Cset$, i.e., always apply the input chosen by the coarsest available controller. The resulting behaviors $\Behtclset$ and $\Behaclset$ are formally defined as in \REFsec{sec:Prelim_FRR} via $\Stset^{cl}$ and $\Saset^{cl}$.

\smallskip
\noindent\textbf{Soundness of Multi-Layered ABCS.}
As shown in \cite{hsu2018multi}, the soundness property of ABCS stated in Prop.~\ref{prop:paimpliespt} transfers to the multi-layered setting.

\begin{proposition}[\cite{hsu2018multi}, Cor.~1]\label{prop:multilayer_soundness}
Let $\Cset$ be a multi-layered controller for the abstract multi-layered system $\Saset$ with effective domains $\Ucset\in\Xaall$ and $\Uc\in X$ inducing the closed loop systems
$\Stclall$ and $\Saclall$, respectively.
Further, let $\Cset\in\WIN(\Saset, \psi)$ for a specification $\psi$ with associated 
behavior $\semantics{\psi}_{\Saset}\subseteq\Beh{\Saset}$ and $\semantics{\psi}_{\Stset}\subseteq\Beh{\Stset}$.
Suppose that for all $\xi\in \Beh{\Stset}$ and $\xia\in\Beh{\Saset}$ s.t.\
(i) $\dom{\xi}=\dom{\xia}$, (ii) for all $k\in \dom{\xi}$ it holds that $(\xi(k),\xia(k))\in \Qset$, and 
(iii) $\xia\in\semantics{\psi}_{\Saset}\Rightarrow\xi\in\semantics{\psi}_{\Stset}$.
Then $\Beh{\Stclall,\Ucset}\subseteq\semantics{\psi}_{\Stset}$, i.e., the time-sampled multi-layered 
closed loop $\Stclall$ fulfills specification $\psi$.
\end{proposition}


\smallskip
\noindent\textbf{Control Problem.}\
Consider the safety and reach-avoid control problems defined over $\Sigma$ in \REFsec{sec:prelim_ABCS}. 
As $\Rset{}$ and $\Tset{}$ can be interpreted as predicates over the state space $X$ of $\Stset{}$, this directly defines the control problems $\tuple{\Stset{},\Specs}$ and $\tuple{\Stset{},\Specr}$ via \eqref{equ:Csound} by substituting $\St{}$ with $\Stset{}$. 

To solve $\tuple{\Stset{},\Spec}$ via multi-layered ABCS we need to ensure that the pre-conditions of Prop.~\ref{prop:multilayer_soundness} hold. This is achieved by \emph{under-approximating} the safe and target sets by a set $\set{\Rseta{l}}_\layer$ and $\set{\Tseta{l}}_\layer$ defined via \eqref{equ:RsetaTseta} for every $\layer$.
Then $\semantics{\Specs}_{\Saset{}}$ and $\semantics{\Specr}_{\Saset{}}$ can be defined via \eqref{equ:Csound} by substituting $\St{}$ with $\Saset{}$, $R$ with $\widehat{R}_{\lambda(\xi(k))}$ and $T$ with $\widehat{T}_{\lambda(\xi(k))}$, where $\lambda(\xa)$ returns the $\layer$ to which $\xa$ belongs, i.e., for $\xa\in\Xa{l}$ we have $\lambda(\xa)=l$. 
We collect all multi-layered controllers $\Cset$ for which $\Beh{\Saclall,\Ucset}\subseteq\semantics{\psi}_{\Saset}$ in $\WIN\tuple{\Saset{},\Spec}$. 
With this, it immediately follows from Prop.~\ref{prop:multilayer_soundness} that $\Cset\in\WIN\tuple{\Saset{},\Spec}$ also solves $\tuple{\Stset{},\Spec}$ via the construction of the time-sampled closed loop system $\Stset^{cl}$.

A multi-layered controller $\Cset\in\WIN(\Saset{},\Spec)$ is typically not unique; there can be many different control strategies implementing the same specification. However, the largest possible controller domain for a particular abstraction layer $l$ always exists and is unique. 
In this paper we will compute a \emph{sound} controller $\Cset\in\WIN(\Saset{},\Spec)$ with a maximal domain w.r.t.\ the lowest layer $l=1$. Formally, for any sound layer $1$ controller $\widetilde{C}=\tuple{\widetilde{\Uc},U,\widetilde{G}}\in\WIN(\Sa{1},\Spec)$ it must hold that $\widetilde{\Uc}$ is contained in the projection 
$\Uc_1=\Qa{1}(\Uc)$ of the effective domain of $\Cset$ to layer $1$. We call such controllers $\Cset$ \emph{complete w.r.t.\ layer $1$}.
On top of that, for faster computation we ensure that cells within its controller domain are only refined if needed.

\section{Controller Synthesis}\label{sec:synthesis}

Our synthesis of an abstract multi-layered controller $\Cset\in\WIN(\Saset{},\Spec)$ has three main ingredients. 
First, we use the usual fixed-point algorithms from reactive synthesis \cite{MPS95} 
to compute the maximal set of winning states 
(i.e., states which can be controlled to fulfill the specification) and deduce an abstract controller
(\REFsec{sec:synthesis:FP}).
Second, we allow switching between abstraction layers during these fixed-point computations by saving and re-loading intermediate results of fixed-point 
computations from and to the lowest layer (\REFsec{sec:synthesis:ML}). 
Third, through the use of \emph{frontiers}, we compute abstractions lazily by only computing abstract transitions in parts of the state space 
currently explored by the fixed-point algorithm (\REFsec{sec:synthesis:Lazy}). 
We prove that frontiers always over-approximate the set of states possibly added to the winning region in the corresponding synthesis step.

\subsection{Fixed-Point Algorithms for Single-Layered ABCS}\label{sec:synthesis:FP}

We first recall the standard algorithms to construct a controller $C$ solving the safety and reach-avoid control problems $\tuple{\Sa{},\Specs}$ and $\tuple{\Sa{},\Specr}$ over the finite abstract system $\Sa{}(\Sigma, \tau,\eta) = (\Xa{},\Ua{},\Fa{})$. 
%
The key to this synthesis is the 
\emph{controllable predecessor} operator, 
$\FCpre{}:\twoup{\Xa{}}\fun\twoup{\Xa{}}$, 
defined for a set $\Aux\subseteq\Xa{}$ by
\begin{equation}
\FCpre{}(\Aux) := \set{\xa\in\Xa{} \mid \exists \ua\in\Ua{}\;.\;\Fa{}(\xa,\ua) \subseteq \Aux}. \label{eq:define cpre}
\end{equation} 
$\tuple{\Sa{},\Specs}$ and $\tuple{\Sa{},\Specr}$ are solved by iterating this operator.


\smallskip
\noindent\textbf{Safety Control.}
Given a safety control problem $\tuple{\Sa{},\Specs}$ associated with $\Rseta{}\subseteq\Ya{}$, one iteratively computes the sets
\begin{equation}\label{equ:safe-fp}
	W^0 = \Rseta{} \text{ and } W^{i+1} = \FCpre{}(W^i) \cap \Rseta{}
\end{equation}
until an iteration $N\in \mathbb{N}$ with $W^N = W^{N+1}$ is reached.
%
From this algorithm, we can extract a safety controller $C = (\Uc,\Ua{},G)$ where $\Uc =W^N$ and
\begin{equation}\label{equ:safe-controller}
\propImp{\ua\in\Gci{}(\xa)}{\Fa{}(\xa,\ua)\subseteq \Uc}
\end{equation}
for all $\xa\in \Uc$.
Note that $C\in\WIN(\Sa{},\Specs)$. 

We denote the procedure implementing this iterative computation until convergence $\SAFE_\infty(\Rseta{},\Sa{})$.
We  also use a version of $\SAFE$ which runs one step of \eqref{equ:safe-fp} only.
Formally, the algorithm $\SAFE(\Rseta{},\Sa{})$ returns the set $W^1$ (the result of the first iteration of \eqref{equ:safe-fp}).
One can obtain $\SAFE_\infty(\Rseta{},\Sa{})$ by chaining $\SAFE$ until convergence, 
i.e., given $W^1$ computed by $\SAFE(\Rseta{},\Sa{})$, one obtains $W^{2}$ from $\SAFE(W^1,\Sa{})$, and so on. 
In \REFsec{sec:synthesis:ML}, we will use such chaining to switch layers after every iteration within our multi-resolution safety fixed-point.

\smallskip
\noindent\textbf{Reach-Avoid Control.}
Given a reach-avoid control problem $\tuple{\Sa{},\Specr}$ for $\Rseta{},\Tseta{}\subseteq\Ya{}$, 
one iteratively computes the sets
\begin{equation}\label{equ:reach-fp}
	W^0 = \Tseta{} \text{ and } W^{i+1} = \left(\FCpre{}(W^i)\cap \Ra{}\right) \cup \Tseta{}
\end{equation}
until some iteration $N\in \mathbb{N}$ is reached where $W^N = W^{N+1}$. 
We extract the reachability controller $C = (\Uc,\Ua{},G)$ with $\Uc = W^N$ and
\begin{equation}\label{equ:reach-controller}
	G(\xa) = \begin{cases}
	          \SetComp{\ua\in \Ua{}}{ \Fa{}(\xa,\ua) \subseteq W^{i*}}, & \xa\in \Uc\setminus \Tseta{}\\
	          \Ua, &\text{else,}
	         \end{cases}
\end{equation}
where $i^* = \min(\{i \mid \xa\in W^i\setminus \Tseta{}\})- 1$.


Note that the safety-part of the specification is taken care of by only keeping those states in $\FCpre{}$ that intersect $\Rseta{}$. 
So, intuitively, the fixed-point in \eqref{equ:reach-fp} iteratively enlarges the target state set while always remaining within the safety constraint. 
We define the procedure implementing the iterative computation of \eqref{equ:reach-fp} until convergence by $\REACH_\infty(\Tseta{},\Rseta{},\Sa{})$. 
We will also use a version of $\REACH$ which runs $m$ steps of \eqref{equ:reach-fp} for a parameter $m\in\mathbb{Z}_{>0}$. 
Here, we can again obtain $\REACH_\infty(\Tseta{},\Rseta{},\Sa{})$ by chaining $\REACH_m$ computations, 
i.e., given $W^m$ computed by $\REACH_m(\Tseta{},\Rseta{},\Sa{})$, one obtains $W^{2m}$ from $\REACH_m(W^m,\Rseta{},\Sa{})$, if no fixed-point is reached beforehand.



\subsection{Multi-Resolution Fixed-Points for Multi-Layered ABCS}\label{sec:synthesis:ML}

Next, we present a controller synthesis algorithm which computes a multi-layered abstract controller $\Cset$ solving the safety and reach-avoid control problems $\tuple{\Saset{},\Specs}$ and $\tuple{\Saset{},\Specr}$ over a sequence of $L$ abstract systems $\Saset:=\{\Sa{l}\}_\layer$. 
Here, synthesis will perform the iterative computations $\SAFE$ and $\REACH$
from \REFsec{sec:synthesis:FP} at each layer, but also switch between abstraction layers during this computation. To avoid notational clutter, we write $\SAFE(\cdot{}, l)$, $\REACH_{\cdot}(\cdot, \cdot, l)$ to refer to
$\SAFE(\cdot, \Sa{l})$, $\REACH_\cdot(\cdot,\cdot, \Sa{l})$ within this procedure.

The core idea that enables switching between layers during successive steps of the fixed-point iterations 
are the saving and re-loading of the computed winning states to and from the lowest layer $l=1$ 
(indicated in green in the subsequently discussed algorithms). 
This projection is formalized by the  operator 
	\begin{equation}\label{equ:Gamma}
	 		\UA{ll'}(\Upsilon_{l'}) = 
			\begin{cases}
				\Ra{ll'}(\Upsilon_{l'}), & l \leq l' \\
				\SetComp{\hat{x} \in \Xa{l}}{\Ra{l'l}(\hat{x})\subseteq \Upsilon_{l'}}, & l > l' 
			\end{cases}
	\end{equation}
where $l,l'\in [1;L]$ and $\Upsilon_{l'}\subseteq \Xa{l'}$. 
The operation $\UA{ll'}(\Upsilon_{l'}) \subseteq \Xa{l}$ under-approximates a set $\Upsilon_{l'} \subseteq \Xa{l'}$
with one in layer $l$. 
 
In this section, we shall assume that each $\Fa{l}$ is pre-computed for all states within $\Rseta{l}$ in every \layer.
In \REFsec{sec:synthesis:Lazy}, we shall compute $\Fa{l}$ lazily.

\smallskip
\noindent\textbf{Safety Control.}
We consider the computation of a multi-layered safety controller $\Cset\in\WIN(\Saset{},\Specs)$ by the iterative function 
$\OTFAS$ in \REFalg{alg:SafeIt} assuming that $\Saset$ is pre-computed. 
We refer to this scenario by the wrapper function $\MLSAFE{}(\Rseta{1},L)$, which calls the iterative algorithm 
$\OTFAS$ with parameters $(\Rseta{1},\emptyset, L, \emptyset)$.
For the moment, assume that the $\COMPUTE$ method in 
line~\ref{alg:SafeIt:Explore} does nothing (i.e., the gray lines of \REFalg{alg:SafeIt} are ignored in the execution).

\begin{algorithm}[t]
	\caption{\textcolor{violet}{$\OTFAS$}}\label{alg:SafeIt}
	\begin{algorithmic}[1]
		\Require $\Psi\subseteq \Xa{1}$, $\Aux\subseteq \Xa{1}$, $l$, $\Cset$
                 \State \textcolor{gray}{$\COMPUTE{}(\UA{l1}(\Psi)\setminus\UA{l1}(\Aux),l)$} \label{alg:SafeIt:Explore}
		\State 	$W \gets \textcolor{blue}{\SAFE}(\textcolor{green!60!black}{\UA{l1}}(\Psi),l)$   \label{alg:SafeIt:computeW}
		\State $\Cset\gets\Cset\cup\{C_l \gets (W,\Ua{},\emptyset)\}$\label{alg:SafeIt:safeB} // store the controller domain, but not moves
		\State $\Aux \gets \Aux\cup\textcolor{green!60!black}{\UA{1l}}(W)$\label{alg:SafeIt:safeW}
		\If{$l\neq 1$} // go finer
			\State $\langle \Psi,\Cset \rangle\gets\textcolor{violet}{\OTFAS}(\Psi,\Aux,l-1,\Cset)$ 
			\State \Return $\langle \Psi,\Cset \rangle$
		\Else
			\If{$\Psi \neq \Aux$}
				\State $\langle \Psi,\Cset \rangle\gets\textcolor{violet}{\OTFAS}(\Aux,\emptyset,L,\emptyset)$ // start new iteration\label{alg:SafeIt:iterate}
				\State \Return $\langle \Psi,\Cset \rangle$
			\Else
				\State \textcolor{purple}{\Return $\langle \Psi,\Cset \rangle$} // terminate \label{alg:SafeIt:termination}
			\EndIf
		\EndIf
	\end{algorithmic}
\end{algorithm}

When initialized with $\OTFAS{}(\Rseta{1},\emptyset,L,\emptyset)$, \REFalg{alg:SafeIt} performs the following computations.
It starts in layer $l=L$ with an outer recursion count $i=1$ (not shown in \REFalg{alg:SafeIt}) 
and reduces $l$, one step at the time, until $l=1$ is reached.
Upon reaching $l=1$, it starts over again from layer $L$ with recursion count $i+1$ and a new safe set $\Aux$. 
In every such iteration $i$, one step of the safety fixed-point is performed for every layer and the resulting set is stored in the layer 
$1$ map $\Aux\subseteq \Xa{1}$, whereas $\Psi\subseteq \Xa{1}$ keeps the knowledge of the previous iteration. 
If the finest layer ($l=1$) is reached and we have $\Psi=\Aux$, the algorithm terminates. 
Otherwise $\Aux$ is copied to $\Psi$, $\Aux$ and $\Cset$ are reset to $\emptyset$ and $\OTFAS{}$ starts a new iteration (see Line~\ref{alg:SafeIt:iterate}).

After $\OTFAS$ has terminated, it returns a multi-layered controller $\Cset=\{\Ci{l}\}_\layer$ (with one controller per layer) which only contains the domains of the respective controllers $\Ci{l}$ (see Line~\ref{alg:SafeIt:safeB} in \REFalg{alg:SafeIt}). The transition functions $\Gci{l}$ are computed afterward by choosing one input $\ua\in\Ua$ for every $\xa\in \Uci{l}$ s.t.\ 
\begin{equation}\label{def:safecontrollerGci}
 \propImp{\ua=\Gci{l}(\xa)}{\Fa{l}(\xa,\ua)\subseteq \UA{l1}(\Psi)}.
\end{equation}

As stated before, the main ingredient for the multi-resolution fixed-point is that states encountered for layer $l$ in iteration $i$ are saved to the lowest layer $1$ (Line~\ref{alg:SafeIt:safeW}, green) and \enquote{loaded} back to the respective layer $l$ in iteration $i+1$ (Line~\ref{alg:SafeIt:computeW}, green). This has the effect that a state $\xa\in\Xa{l}$ with $l>1$, which was not contained in $W$ computed in layer $l$ and iteration $i$ via Line~\ref{alg:SafeIt:computeW}, might be included in $\UA{l1}(\Psi)$ loaded in the next iteration $i+1$ when re-computing Line~\ref{alg:SafeIt:computeW} for $l$. This happens if all states $x\in\xa$ were added to $\Aux$ by some layer $l'<l$ in iteration $i$. 
%
%

Due to the effect described above, the map $W$ encountered in Line~\ref{alg:SafeIt:computeW} for a particular layer $l$ throughout different iterations $i$ might not be monotonically shrinking. However, the latter is true for layer $1$, which implies that $\MLSAFE{}(\Rseta{1},L)$ is sound and complete w.r.t.\ layer $1$ as formalized by \REFthm{thm:EagerSafe}. 



\begin{theorem}[\cite{hsu2018lazy}]\label{thm:EagerSafe}
	$\MLSAFE{}$ is sound and complete w.r.t.\ layer $1$.
\end{theorem}

It is important to mention that the algorithm $\MLSAFE{}$ is presented only to make a smoother transition to the lazy ABCS for safety (to be presented in the next section).
In practice, $\MLSAFE{}$ itself is of little algorithmic value as it is always slower than $\SAFE(\cdot, \Sa{1})$, but produces the same result.
This is because in $\MLSAFE{}$, the fixed-point computation in the finest layer does not use the coarser layers' winning domain in any meaningful way.
So the computation in all the layers---except in $\Sa{1}$---goes to waste.

\smallskip
\noindent\textbf{Reach-Avoid Control.}
We consider the computation of an abstract multi-layered reach-avoid controller $\Cset\in\WIN(\Saset{},\Specr)$ by the iterative function 
$\OTFAR$ in \REFalg{alg:OTFA subroutine} assuming that $\Saset$ is pre-computed.
We refer to this scenario by the wrapper function $\mathsf{EagerReach}(\Tseta{1},\Rseta{1},L)$, which calls 
$\OTFAR$ with parameters 
$(\Tseta{1}, \Rseta{1}, L, \emptyset)$.
%
Assume in this section that $\COMPUTE$ and $\EXPLORE_m$ do not modify anything (i.e., the gray lines of \REFalg{alg:OTFA subroutine} are ignored in the execution).

\begin{algorithm}[t!]
	\caption{\textcolor{violet}{$\OTFAR_{m}$}}\label{alg:OTFA subroutine}
	\begin{algorithmic}[1]
		\Require $\Aux\subseteq \Xa{1}$,$\Psi\subseteq \Xa{1}$, $l$, $\Cset$
		\If{$l=L$}\label{alg:if-start}
		        \State \textcolor{gray}{$\COMPUTE{}(\UA{l1}(\Psi),l)$} \label{alg:Reach:Explore}
			\State $\langle W,C \rangle \gets \textcolor{blue}{\REACH_\infty}(\textcolor{green!60!black}{\UA{l1}}(\Aux),\UA{l1}(\Psi),l)$   \label{alg:full reach}
			\State $\Cset \gets \Cset \cup \set{C}$ 	
			\State $\Aux \gets \Aux \cup \textcolor{green!60!black}{\UA{1l}}(W)$ // save $W$ to $\Aux$ \label{alg:safe1}
			\If{$L=1$} \qquad // single-layered reachability
				\State \textcolor{purple}{\Return $\langle \Aux,\Cset \rangle$} \label{alg:L=1}
			\Else \qquad // go finer
				\State $\langle \Aux,\Cset \rangle \gets \textcolor{violet}{\OTFAR_{m}}\left(\Aux,\Psi,l-1,\Cset\right)$ \label{alg:go finer 1}
				\State \Return $\langle \Aux,\Cset \rangle$
			\EndIf
		\Else\label{alg:else-start}
		        \State \textcolor{gray}{$\EXPLORE_m(\Aux,l)$}
			\State $\langle W,C \rangle \gets \textcolor{blue}{\REACH_m}(\textcolor{green!60!black}{\UA{l1}}(\Aux),\UA{l1}(\Psi),l)$ \label{alg:compute reach}
			\State $\Cset \gets \Cset \cup \set{C}$ 
			\State $\Aux \gets \Aux \cup \textcolor{green!60!black}{\UA{1l}}(W)$ // save $W$ to $\Aux$ \label{alg:safe2}
			\If{Fixed-point is reached in line \ref{alg:compute reach}} \label{alg:checkFP}
				\If{$l = 1$} // finest layer reached \label{alg:terminate}
					\State \textcolor{purple}{\Return $\langle \Aux,\Cset \rangle$}\label{alg:l=1}
				\Else \qquad // go finer
					\State $\langle \Aux,\Cset \rangle \gets \textcolor{violet}{\OTFAR_{m}}(\Aux,\Psi,l-1,\Cset)$	\label{alg:go finer 2}
					\State \Return $\langle \Aux,\Cset \rangle$
				\EndIf
			\Else \qquad // go coarser \label{alg:go coarser begin}
				\State $\langle \Aux,\Cset \rangle \gets \textcolor{violet}{\OTFAR_{m}}(\Aux,\Psi,l+1,\Cset)$	\label{alg:go coarser}
				\State \Return $\langle \Aux,\Cset \rangle$	\label{alg:go coarser end}
			\EndIf
		\EndIf\label{alg:else-end}
	\end{algorithmic}
\end{algorithm}

The recursive procedure $\OTFAR_m$ in \REFalg{alg:OTFA subroutine} implements the switching protocol informally discussed in \REFsec{sec:Introduction}.
Lines \ref{alg:if-start}--\ref{alg:else-start} implement the fixed-point computation at the coarsest layer $\Sa{L}$ 
by iterating the fixed-point over $\Sa{L}$ until convergence (line \ref{alg:full reach}). 
Afterward, $\OTFAR_m$ recursively calls itself (line \ref{alg:go finer 1}) 
to see if the set of winning states ($W$) can be extended by a lower abstraction layer.
Lines \ref{alg:else-start}--\ref{alg:else-end} implement the fixed-point computations in layers $l < L$ 
by iterating the fixed-point over $\Sa{l}$ for $m$ steps (line \ref{alg:compute reach}) 
for a given fixed parameter $m>0$.
If the analysis already reaches a fixed point, then, as in the first case, the algorithm $\OTFAR_m$ recursively calls itself 
(line \ref{alg:go finer 2}) to check if further states can be added in a lower layer.
If no fixed-point is reached in line \ref{alg:compute reach}, more states could be added in the current layer by running $\REACH$ 
for more then $m$ steps. 
However, this might not be efficient (see the example in \REFsec{sec:Introduction}). 
The algorithm therefore attempts to go coarser when recursively calling itself (line~\ref{alg:go coarser}) 
to expand the fixed-point in a coarser layer instead. 
Intuitively, this is possible if states added by lower layer fixed-point computations have now \enquote{bridged} 
a region where precise control was needed and can now be used to enable control in coarser layers again. 
This also shows the intuition behind the parameter $m$. 
If we set it to $m=1$, the algorithm might attempt to go coarser before this \enquote{bridging} is completed. 
The parameter $m$ can therefore be used as a tuning parameter to adjust the frequency of such attempts and is only needed in layers $l<L$.
The algorithm terminates if a fixed-point is reached in the lowest layer (line \ref{alg:L=1} and line \ref{alg:l=1}). 
In this case the layer $1$ winning state set $\Aux$ and the multi-layered controller $\Cset$ is returned.

It was shown in \cite{hsu2018multi} that this switching protocol ensures that $\mathsf{EagerReach}_m$ is 
sound and complete w.r.t.\ layer $1$. 

\begin{theorem}[\cite{hsu2018multi}]\label{thm:EagerReach}
	$\mathsf{EagerReach}_m$ is sound and complete w.r.t.\ layer $1$.
\end{theorem}


\vspace*{-0.5cm}
\noindent\begin{minipage}[t]{0.45\linewidth}
\begin{algorithm}[H]
	\caption{$\COMPUTE{}$}\label{alg:Compute}
	\begin{algorithmic}[1]
		\Require $\Aux\subseteq \Xa{l}$, $l$
		\For {$\xa\in\Aux,\ua\in\Ua$}
		\If {$\Fa{l}(\xa,\ua)$ is undefined}
		\State compute $\Fa{l}(\xa,\ua)$ as in \eqref{eq:next state abs sys 0}
		\EndIf
		\EndFor
	\end{algorithmic}
\end{algorithm}
\end{minipage}
\hfill
\begin{minipage}[t]{0.45\linewidth}
\begin{algorithm}[H]
	\caption{$\EXPLORE_m$}\label{alg:EXPLORE subroutine}
	\begin{algorithmic}[1]
		\Require $\Aux\subseteq \Xa{1}$, $l$
		\State $W' \gets \JUpre{l}^{m}(\OA{L1}(\Aux)) \setminus \UA{L1} (\Aux)$\label{alg:overapprox states} 
		\State $W'' \gets \UA{lL}(W')$ \label{alg:candidate states for explore}
		\State $\COMPUTE{}(W''\cap\Rseta{l},l)$\label{alg:partial abs}
	\end{algorithmic}
\end{algorithm}
\end{minipage}

\subsection{Lazy Exploration within Multi-Layered ABCS}\label{sec:synthesis:Lazy}

We now consider the case where the multi-layered abstractions $\Saset$ are computed lazily. 
Given the multi-resolution fixed-points discussed in the previous section, this 
requires tightly over-approximating the region of the state space which might be explored by $\REACH$ or $\SAFE$ 
in the current layer, called the \emph{frontier}. 
Then abstract transitions are only constructed for \emph{frontier states} and the currently considered layer $l$ via 
\REFalg{alg:Compute}.
As already discussed in \REFsec{sec:Introduction}, the computation of \emph{frontier states} differs for safety 
and reachability objectives.

\smallskip
\noindent\textbf{Safety Control.}
We now consider the \emph{lazy} computation of a multi-layered safety controller $\Cset\in\WIN(\Saset,\Specs)$. 
We refer to this scenario by the wrapper function $\MLASAFE(\Rseta{1},L)$ which simply calls $\OTFAS{}(\Rseta{1},\emptyset,L,\emptyset)$.

This time, Line~\ref{alg:SafeIt:Explore} of \REFalg{alg:SafeIt} is used to explore transitions.
The frontier cells at layer $l$ are given by $\mathcal{F}_l=\UA{l1}(\Psi)\setminus\UA{l1}(\Aux)$.
The call to $\COMPUTE{}$ in \REFalg{alg:Compute} updates the abstract transitions for the frontier cells.
In the first iteration of $\OTFAS{}(\Rseta{1},\emptyset,L,\emptyset)$, 
we have $\Psi=\Rseta{1}$ and $\Aux=\emptyset$.
Thus, $\mathcal{F}_L=\UA{1L}(\Rseta{1})=\Rseta{L}$, and hence, for layer $L$, 
all transitions for states inside the safe set are pre-computed in the first iteration of \REFalg{alg:SafeIt}.
In lower layers $l<L$, the frontier $\mathcal{F}_l$ defines all states which are 
(i) not marked unsafe by all layers in the previous iteration, i.e., are in $\UA{l1}(\Psi)$, but 
(ii) cannot stay safe for $i$ time-steps in any layer $l'>l$, 
i.e., are not in $\UA{l1}(\Aux)$. 
Hence, $\mathcal{F}_l$ defines a small boundary around the set $W$ computed in 
the previous iteration of $\SAFE$ in layer $l+1$ 
(see \REFsec{sec:Introduction} for an illustrative example of this construction). 

It has been shown in \cite{hsu2018lazy} that  all states which need to be checked for safety in layer $l$ 
of iteration $i$ are indeed explored by this frontier construction. This implies that \REFthm{thm:EagerSafe} directly transfers from $\MLSAFE$ to $\MLASAFE$.

\begin{theorem}
	$\MLASAFE$ is sound and complete w.r.t.\ layer $1$.
\end{theorem}

\smallskip
\noindent\textbf{Reach-Avoid Control.}
We now consider the lazy computation of a multi-layered reach-avoid controller $\Cset\in\WIN(\Saset,\Specr)$. 
We refer to this scenario by the wrapper function $\MLAR_m(\Tseta{1},\Rseta{1},L)$ which calls $\OTFAR_{m}(\Tseta{1},\Rseta{1},L,\emptyset)$.

In the first iteration of $\OTFAR_{m}$ we have the same situation as in $\MLASAFE$; given that $\Psi=\Rseta{1}$, line~\ref{alg:Reach:Explore} in \REFalg{alg:OTFA subroutine} pre-computes all transitions for states inside the safe set and $\COMPUTE{}$ does not perform any computations for layer $L$ in further iterations. For $l<L$ however, the situation is different. As $\REACH$ computes a smallest fixed-point, it iteratively enlarges the set $\Tseta{1}$ (given when $\OTFAR$ is initialized). Computing transitions for all not yet explored states in every iteration would therefore be very wasteful (see the example in \REFsec{sec:Introduction}). 
Therefore, $\EXPLORE_m$ determines an over-approximation of the frontier states instead in the following manner: 
it computes the predecessors (not controllable predecessors!)
of the already-obtained set $\Aux$ optimistically by (i) using (coarse) auxiliary abstractions for this computation and 
(ii) applying a \emph{cooperative predecessor} operator.

This requires a set of auxiliary systems, given by
\begin{equation}
 \Aaset = \{\Aa{l}\}_{l=1}^L,\qquad\Aa{l}:=\Sa{}(\Sigma,\tau_l,\eta_L) = (\Xa{L},\Ua,\Ja{l}).
\end{equation}
The abstract system $\Aa{l}$ induced by $\Sigma$ captures the $\tau_l$-duration transitions in 
the coarsest layer state space $\Xa{L}$.
Using $\tau_l$ instead of $\tau_L$ is important, 
as $\tau_L$ might cause \enquote{holes} between the computed frontier 
and the current target $\Upsilon$ which cannot be bridged by a shorter duration control actions in layer $l$.
This would render $\MLAR_{m}$ unsound.
Also note that in $\EXPLORE_m$, we do not restrict attention to the safe set. 
This is because $\Ra{l}\supseteq\Ra{L}$, and when the inequality is strict then 
the safe states in layer $l$ which are possibly winning but are covered 
by an obstacle in layer $L$ (see \REFfig{fig:informal example}) can also be explored.

For $\Aux\subseteq\Xa{L}$ and $l\in[1;L]$, we define the \emph{cooperative predecessor} operator
\begin{align}
\JUpre{l}(\Aux)= \set{\xa\in\Xa{L} \mid \exists \ua\in\Ua{}\;.\;\Ja{l}(\xa,\ua)\cap \Aux \neq \emptyset}. \label{eq:define upre}
\end{align}
in analogy to the controllable predecessor operator in \eqref{eq:define cpre}.
We apply the cooperative predecessor operator $m$ times in $\EXPLORE_m$, i.e.,
\begin{align*}\label{eq:def upre^r}
&\JUpre{l}^1(\Aux)  = \JUpre{l}(\Aux)~\text{and}~\\ 
&\JUpre{l}^{j+1}(\Aux)  = \JUpre{l}^{j}(\Aux) \cup \JUpre{l}(\JUpre{l}^j(\Aux)). \numberthis
\end{align*}
Calling $\EXPLORE_m$ with parameters $\Aux\subseteq \Xa{1}$ and $l<L$ applies $\JUpre{l}^m$ to the over-approximation of $\Aux$ by abstract states in layer $L$. 
This over-approximation is defined as the dual operator of the under-approximation operator $\UA{ll'}$:
\begin{align}\label{eq:def over approx}
	\OA{ll'}(\Upsilon_{l'}) := \begin{cases}
								\Ra{ll'}(\Upsilon_{l'}), & l \leq l' \\
								\SetComp{\hat{x} \in \Xa{l}}{\Ra{l'l}(\hat{x})\cap \Upsilon_{l'} \neq \emptyset}, & l > l' 
							\end{cases}
\end{align}
where $l,l'\in [1;L]$ and $\Upsilon_{l'}\subseteq \Xa{l'}$. 
Finally, $m$ controls the size of the frontier set and determines the maximum progress that can be made in a single backwards synthesis run in a layer $l < L$.

It can be shown 
that all states which might be added to the winning state set in the 
current iteration are indeed explored by this frontier construction, implying that 
$\MLAR_m(\Tseta{1},\Rseta{1},L)$ is sound and complete w.r.t.\ layer $1$. 
In other words, \REFthm{thm:EagerReach} can be transfered from $\MLREACH_m$ to $\MLAR_{m}$.

\begin{theorem}\label{thm:s&c lazyreach}
	$\MLAR_m$ is sound and complete w.r.t.\ layer $1$.
\end{theorem}

\section{Experimental Evaluation}

We have implemented our algorithms in the MASCOT tool and we present some brief evaluation.\footnote{
Available at \url{http://mascot.mpi-sws.org/}.}

\subsection{Reach-Avoid Control Problem for a Unicycle}
%
\begin{figure}[t]
  \begin{tikzpicture}[auto,scale=1]
    \node[inner sep=0pt] (example) at (0,0) {\includegraphics[width=.49\textwidth]{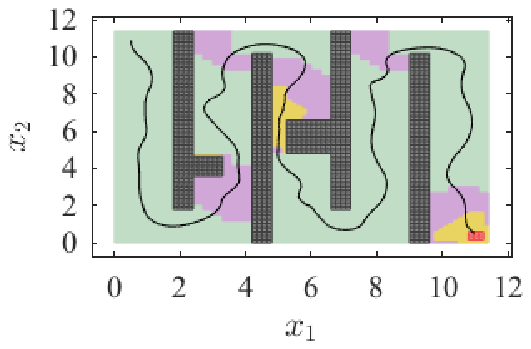}};
   \node (a) at ($(example.north west)+(0.1,-0.3)$) {$(a)$};
  
  \end{tikzpicture}\hfill
    \begin{tikzpicture}[auto,scale=1]
    \node[inner sep=0pt] (example) at (0,0) {\includegraphics[width=.49\textwidth]{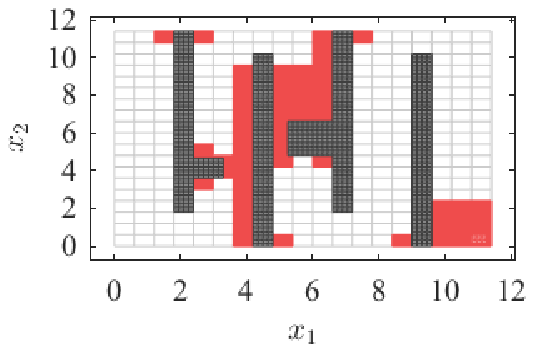}};
   \node (a) at ($(example.north west)+(0.1,-0.3)$) {$(b)$};
  
  \end{tikzpicture}
	\vspace{-0.5cm}
	\caption{(a) Solution of the unicycle reach-avoid problem by $\MLAR_{2}$. (b) Cells of the finest layer ($l=1$) for which transitions were computed during $\MLAR_{2}$ are marked in red. For $\MLREACH_{2}$, all uncolored cells would also be red. 
	}
\label{fig:hsccexample}
\end{figure}

We use a nonlinear kinematic system model commonly known as the \textit{unicycle model}, specified as
\begin{align*}
	\dot{x}_1 \in u_1\cos(x_3)	+ W_1\quad
	\dot{x}_2 \in u_1\sin(x_3)	+ W_2\quad
	\dot{x}_3 = u_2
\end{align*}
where $x_1$ and $x_2$ are the state variables representing 2D Cartesian coordinates, $x_3$ is
a state variable representing the angular displacement, 
$u_1$, $u_2$ are control input variables that influence the linear and angular velocities respectively, 
and $W_1$, $W_2$ are the perturbation bounds in the respective dimensions 
given by $W_1 = W_2 = [-0.05, 0.05]$. 
The perturbations render this deceptively simple problem computationally intensive. 
We run controller synthesis experiments for the unicycle inside a two dimensional space with obstacles and a 
designated target area, as shown in \REFfig{fig:hsccexample}.
We use three layers for the multi-layered algorithms $\MLR$ and $\MLAR$.
All experiments presented in this subsection were performed on a Intel Core i5 3.40 GHz processor.

\smallskip
\noindent\textbf{Algorithm Comparison.}
Table~\ref{table:REACH} shows a comparison on the $\REACH$, $\MLR_2$, and $\MLAR_{2}$ algorithms. 
The projection to the state space of the transitions constructed by $\MLAR_{2}$ for the finest abstraction is shown in 
\REFfig{fig:hsccexample}b. 
The corresponding visualization for $\MLR_{2}$ would show all of the uncolored space being covered by red. 
The savings of $\MLAR_{2}$ over $\MLR_2$ can be mostly attributed to this difference.


\vspace*{-0.3cm}
\noindent\begin{minipage}[t]{0.5\linewidth}
\begin{table}[H]
	\centering
	\caption{Comparison of running times (in seconds) of reachability algorithms on the perturbed unicycle system.}
	\label{table:REACH}
	\begin{scriptsize}
		\begin{tabular}{lrrrr} \toprule
			
			& $\quad \REACH$ & $\quad \MLREACH_2$ & $\quad \MLAR_{2}$ \\ \midrule
			
			Abstraction & 2590 & 2628 & 588 \\
			Synthesis & 818 & 73  & 21 \\
			
			\midrule
			Total &  3408 & 2701 & 609  \\
			& (126\%)  & (100\%) & (22.5\%)\\
			
			\bottomrule		
		\end{tabular}
	\end{scriptsize}
\end{table}
\end{minipage}
\hfill
\begin{minipage}[t]{0.45\linewidth}
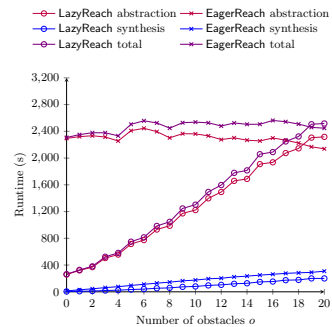
\begin{figure}[H]
\begin{tikzpicture}[scale=0.5,point/.style={circle,scale=0.2,draw=black,fill=red,thick}]	
	\begin{axis}[
		xlabel= Number of obstacles $o$,
		ylabel= Runtime (s),
		axis y line=left,
		axis x line=bottom,
		xmin=0, xmax=20,
		ymin=0, ymax=3200,
		xtick={0,2,4,6,8,10,12,14,16,18,20},
		ytick={0,400,800,1200,1600,2000,2400,2800,3200},
		xticklabel=$\pgfmathprintnumber{\tick}$,
		yticklabel=$\pgfmathprintnumber{\tick}$,
		legend columns=2,
		legend style={at={(-0.15, 1.35)},anchor=north west,draw=none},
		legend cell align={left}
        ]
        \addplot[purple, mark=o]  
        	coordinates {
        		(0, 258.7)
        		(1, 319.9)
        		(2, 369.2)
        		(3, 504.9)
        		(4, 556.9)
        		(5, 715.4)
        		(6, 773.8)
        		(7, 931.7)
        		(8, 987.4)
        		(9, 1172.8)
        		(10, 1222.9)
        		(11, 1397.9)
        		(12, 1492.6)
        		(13, 1659.3)
        		(14, 1688.5)
        		(15, 1909.87)
        		(16, 1936.44)
        		(17, 2074.58)
        		(18, 2145.23)
        		(19, 2305.66)
        		(20, 2317.3)
        	};
			\addlegendentry{$\MLAR$ abstraction}
		\addplot[purple, mark=x]  
            coordinates {
            	(0, 2294.42)
            	(1, 2320.97)
            	(2, 2334.09)
            	(3, 2315.2)
            	(4, 2256.18)
            	(5, 2410.43)
            	(6, 2445.3)
            	(7, 2396.04)
            	(8, 2303.2)
            	(9, 2365.08)
            	(10, 2361.76)
            	(11, 2332.85)
            	(12, 2276.78)
            	(13, 2302.51)
            	(14, 2267.68)
            	(15, 2256.3)
            	(16, 2300.15)
            	(17, 2267.99)
            	(18, 2224.44)
            	(19, 2168.56)
            	(20, 2138.86)
            };
            \addlegendentry{$\MLR$ abstraction}
        \addplot[blue, mark=o]
        	coordinates {
        		(0, 4.6)
        		(1, 7.94)
        		(2, 11.7)
        		(3, 19.4)
        		(4, 24.3)
        		(5, 34.1)
        		(6, 41.0)
        		(7, 53.3)
        		(8, 58.8)
        		(9, 74.0)
        		(10, 80.5)
        		(11, 95.6)
        		(12, 104.2)
        		(13, 119.8)
        		(14, 127.7)
        		(15, 148.2)
        		(16, 153.5)
        		(17, 174.055)
        		(18, 177.763)
        		(19, 200.539)
        		(20, 200.08)
        	};
        	\addlegendentry{$\MLAR$ synthesis};
        \addplot[blue, mark=x]
            coordinates {
            	(0, 13.1959)
            	(1, 30.2978)
            	(2, 45.0114)
            	(3, 62.825)
            	(4, 76.63)
            	(5, 95.1656)
            	(6, 112.505)
            	(7, 129.46)
            	(8, 144.514)
            	(9, 163.755)
            	(10, 175.585)
            	(11, 194.599)
            	(12, 202.9)
            	(13, 222.501)
            	(14, 234.388)
            	(15, 249.656)
            	(16, 262.43)
            	(17, 275.995)
            	(18, 284.01)
            	(19, 290.98)
            	(20, 308.32)
            };
            \addlegendentry{$\MLR$ synthesis};
        \addplot[violet, thick, mark=o]
        	coordinates { 
        		(0, 263.3)
        		(1, 327.8)
        		(2, 380.9)
        		(3, 524.3)
        		(4, 581.2)
        		(5, 749.5)
        		(6, 814.8)
        		(7, 985.0)
        		(8, 1046.8)
        		(9, 1246.8)
        		(10, 1303.4)
        		(11, 1493.5)
        		(12, 1596.8)
        		(13, 1779.0)
        		(14, 1816.2)
        		(15, 2058.1)
        		(16, 2089.9)
        		(17, 2248.6)
        		(18, 2322.99)
        		(19, 2506.2)
        		(20, 2517.38)
        	};
        	\addlegendentry{$\MLAR$ total};
        \addplot[violet, thick, mark=x]
            coordinates { 
				(0, 2307.62)
				(1, 2351.27)
				(2, 2379.10)
				(3, 2378.025)
				(4, 2332.81)
				(5, 2505.60)
				(6, 2557.81)
				(7, 2525.5)
				(8, 2447.714)
				(9, 2528.84)
				(10, 2537.35)
				(11, 2527.449)
				(12, 2479.68)
				(13, 2525.01)
				(14, 2505.068)
				(15, 2505.96)
				(16, 2562.58)
				(17, 2543.99)
				(18, 2508.45)
				(19, 2459.54)
				(20, 2447.18)
            };
            \addlegendentry{$\MLR$ total};
	\end{axis}
\end{tikzpicture}	
\vspace*{-0.2cm}
\caption{Runtime with increasing number of obstacles}
\label{fig:obstacles}
\end{figure}
\end{minipage}





\vspace{0.2cm}
\smallskip
\noindent\textbf{Varying State Space Complexity.}
We investigate how the lazy algorithm and the multi-layered baseline perform with respect to the structure of the state space,
achieved by varying the number of identical obstacles, $o$, placed in the open area of the state space. 
The runtimes for $\MLR_2$ and $\MLAR_{2}$ are plotted in \REFfig{fig:obstacles}. 
We observe that $\MLAR_{2}$ runs fast when there are few obstacles by only constructing the 
abstraction in the finest layer for the immediate surroundings of those obstacles. 
By $o = 20$, $\MLAR_{2}$ explores the entire state space in the finest layer, and its 
performance is slightly worse than that of $\MLR_2$ (due to additional bookkeeping). 
The general decreasing trend in the abstraction construction runtime for $\MLR_2$ 
is because transitions outgoing from obstacle states are not computed.


%

\subsection{Safety Control Problem for a DC-DC Boost Converter \cite{hsu2018lazy}} \label{sec:experiments}

We evaluate our safety algorithm on a benchmark DC-DC boost converter example from \cite{GirardPolaTabuada_2010,mouelhi2013cosyma,Scots}. 
%
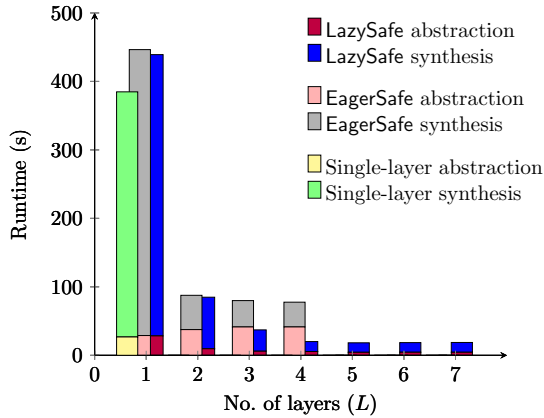
\begin{figure}[t]
	\centering
	\begin{tikzpicture}[scale=.8,point/.style={circle,scale=0.2,draw=black,fill=red,thick}]	
	\begin{axis}[
	    	legend cell align={left},
		ybar stacked,
		bar shift=3pt,
		xlabel= No. of layers ($L$),
		ylabel= Runtime (s),
		axis y line=left,
		axis x line=bottom,
		xmin=0, xmax=8,
		ymin=0, ymax=500,
		xtick={0,1,2,3,4,5,6,7},
		ytick={0,100,200,300,400,500},
		xticklabel=$\pgfmathprintnumber{\tick}$,
		yticklabel=$\pgfmathprintnumber{\tick}$,
		legend style={at={(0.5,1)},anchor=north west,draw=none},
        ]
        \addplot[fill=purple]	
			coordinates {
        		(1, 28.355)
        		(2, 9.59962)
        		(3, 5.98247)
        		(4, 5.54938)
        		(5, 4.39672)
        		(6, 4.48306)
        		(7, 4.45758)	
        	};
			\addlegendentry{$\MLASAFE{}$ abstraction}
		
        \addplot[fill=blue]	
			coordinates {
        		(1, 410.783)
        		(2, 75.2443)
        		(3, 31.1342)
        		(4, 14.5017)
        		(5, 13.8379)
        		(6, 14.0392)
        		(7, 14.2794)	
        	};
        	\addlegendentry{$\MLASAFE{}$ synthesis};
        
	\end{axis}
	
	\begin{axis}[
	    	legend cell align={left},
		ybar stacked,
		bar shift=-3pt,
		xlabel= No. of layers ($L$),
		ylabel= Runtime (s),
		axis y line=left,
		axis x line=bottom,
		xmin=0, xmax=8,
		ymin=0, ymax=500,
		xtick={0,1,2,3,4,5,6,7},
		ytick={0,100,200,300,400,500},
		xticklabel=$\pgfmathprintnumber{\tick}$,
		yticklabel=$\pgfmathprintnumber{\tick}$,
		legend style={at={(0.5,0.8)},anchor=north west,draw=none},
        ]
        
        \addplot[fill=red!30]	
			coordinates {
        		(1, 28.8522)
        		(2, 37.5347)
        		(3, 41.3946)
        		(4, 41.415)
        		(5, 0)
        		(6, 0)
        		(7, 0)	
        	};
            \addlegendentry{$\MLSAFE{}$ abstraction}
            
        \addplot[fill=black!30]	
			coordinates {
        		(1, 417.597)
        		(2, 50.0761)
        		(3, 38.5002)
        		(4, 36.1657)
        		(5, 0)
        		(6, 0)
        		(7, 0)	
        	};
            \addlegendentry{$\MLSAFE{}$ synthesis};
	\end{axis}
	
	\begin{axis}[
	    	legend cell align={left},
		ybar stacked,
		bar shift=-9pt,
		xlabel= No. of layers ($L$),
		ylabel= Runtime (s),
		axis y line=left,
		axis x line=bottom,
		xmin=0, xmax=8,
		ymin=0, ymax=500,
		xtick={0,1,2,3,4,5,6,7},
		ytick={0,100,200,300,400,500},
		xticklabel=$\pgfmathprintnumber{\tick}$,
		yticklabel=$\pgfmathprintnumber{\tick}$,
		legend style={at={(0.5,0.6)},anchor=north west,draw=none},
        ]
        \addplot[fill=yellow!50]	
			coordinates {
        		(1, 26.7736)
        		(2, 0)
        		(3, 0)
        		(4, 0)
        		(5, 0)
        		(6, 0)
        		(7, 0)	
        	};
			\addlegendentry{Single-layer abstraction}
		
        \addplot[fill=green!50]	
			coordinates {
        		(1, 357.999)
        		(2, 0)
        		(3, 0)
        		(4, 0)
        		(5, 0)
        		(6, 0)
        		(7, 0)	
        	};
        	\addlegendentry{Single-layer synthesis};
        
	\end{axis}
\end{tikzpicture}	
	\caption{Run-time comparison of $\MLASAFE{}$ and $\MLSAFE{}$ on the DC-DC boost converter example. $L > 4$ is not used for $\MLSAFE{}$ since coarser layers fail to produce a non-empty winning set. The same is true for $L > 7$ in $\MLASAFE{}$.}	
	\label{fig:runtimes}
	\vspace{-0.5cm}
\end{figure}
%
The system $\Sigma$ is a second order differential inclusion $\dot{X}(t) \in A_pX(t) + b + W$ with two switching modes $p\in \set{1,2}$, where
\begin{align*}
	b = \begin{bmatrix}
		\frac{v_s}{x_l}\\
		0
	\end{bmatrix}, 
	A_1 = \begin{bmatrix}
		-\frac{r_l}{x_l}  &  0\\
		0 				  &  -\frac{1}{x_c}\frac{r_0}{r_0+r_c}
	\end{bmatrix},
	A_2 = \begin{bmatrix}
		-\frac{1}{x_l}(r_l+\frac{r_0r_c}{r_0+r_c})	&	\frac{1}{5}(-\frac{1}{x_l}\frac{r_0}{r_0+r_c})\\
		5\frac{r_0}{r_0+r_c}\frac{1}{x_c}			&	-\frac{1}{x_c}\frac{1}{r_0+r_c}
	\end{bmatrix},
\end{align*}
with $r_0 = 1$, $v_s = 1$, $r_l = 0.05$, $r_c = 0.5r_l$, $x_l = 3$, $x_c = 70$ and $W = [-0.001, 0.001]\times [-0.001, 0.001]$. 
A physical and more detailed description of the model can be found in \cite{GirardPolaTabuada_2010}. 
The safety control problem that we consider is given by $\langle \Sigma,\Specs \rangle$, where $\Specs = always( [1.15,1.55]\times [5.45,5.85])$. 
We evaluate the performance of our $\MLASAFE{}$ algorithm on this benchmark and compare it to $\MLSAFE{}$ and a single-layered baseline. 
For $\MLASAFE{}$ and $\MLSAFE{}$, we vary the number of layers used. 
The results are presented in \REFfig{fig:runtimes}.
In the experiments, the finest layer is common, and is parameterized by $\eta_1 = [0.0005, 0.0005]$ and $\tau_1 = 0.0625$.
The ratio between the grid parameters and the sampling times of the successive layers is $2$.

From \REFfig{fig:runtimes}, we see that $\MLASAFE{}$ is significantly faster than both $\MLSAFE{}$ (and the single-layered baseline) as $L$ increases.
The single layered case ($L=1$) takes slightly more time in both $\MLASAFE{}$ and $\MLSAFE{}$ due to the extra bookkeeping in the multi-layered algorithms. 
In \REFfig{fig:dcdc}, we visualize the domain of the constructed transitions and the synthesized controllers in each layer for $\MLASAFE{}(\cdot,6)$. 
The safe set is mostly covered by cells in the two coarsest layers. 
This phenomenon is responsible for the computational savings over $\MLASAFE{}(\cdot,1)$.

In contrast to the reach-avoid control problem for a unicycle, in this example, synthesis takes significantly longer time than the abstraction.
To reason about this difference is difficult, because the two systems are completely incomparable, and the abstraction parameters are very different.
Still we highlight two suspected reasons for this mismatch:
\begin{inparaenum}[(a)]
	\item Abstraction is faster because of the lower dimension and smaller control input space of the boost converter, 
	\item A smaller sampling time ($0.0625s$ as compared to $0.225s$ for the unicycle) in the 
	finest layer of abstraction for the boost converter results in slower convergence of the fixed-point iteration.
\end{inparaenum}

\begin{figure}[t]
	\centering
	\includegraphics[width=1\columnwidth,scale=0.5]{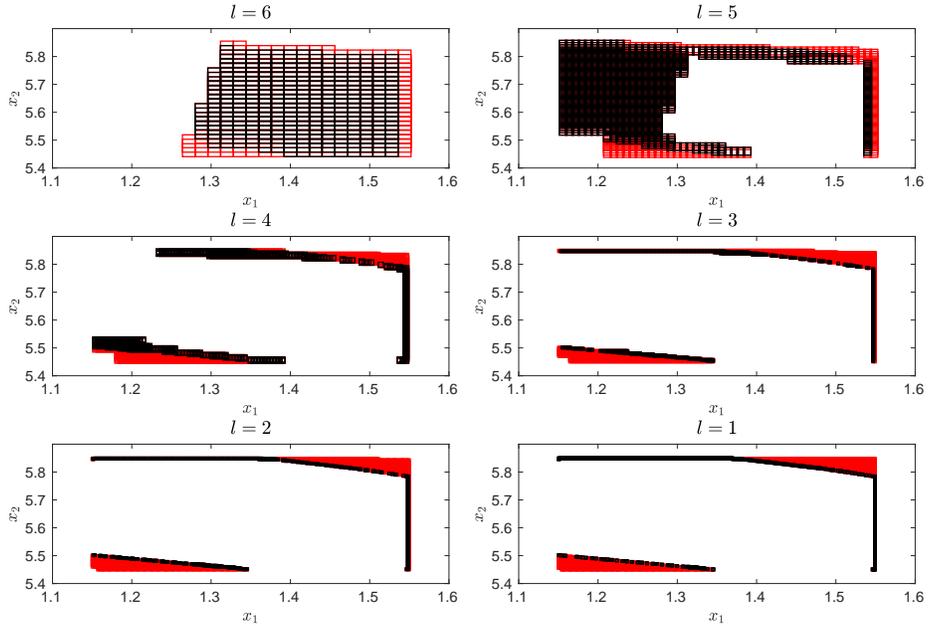}
	\vspace{-0.3cm}
	\caption{Domain of the computed transitions (union of red and black region) and the synthesized controllers (black region) for the DC-DC boost converter example, computed by $\MLASAFE{}(\cdot,6)$.}
	\label{fig:dcdc}
\end{figure}

\section{Conclusion}

ABCS is an exciting new development in the field of formal synthesis of cyber-physical systems. 
We have summarized a multi-resolution approach to ABCS.  Fruitful avenues for future work include designing scalable and robust tools, and combining basic algorithmic techniques with structural heuristics or orthogonal techniques (e.g. those based on data-driven exploration).

\bibliographystyle{abbrv}
{
\bibliography{reportbib}}
%
\section*{Appendix: Proof of \REFthm{thm:s&c lazyreach}}
First we state some properties of $\UA{ll'}(\cdot)$ and $\OA{ll'}(\cdot)$. Let $A_{l'}, B_{l'}\subseteq \Xa{l'}$ be any two sets. Then\\
\begin{inparaenum}[(a)]
	\item $\OA{ll'}(\cdot) = \OA{lk}(\OA{kl'}(\cdot))$ and $\UA{ll'}(\cdot) = \UA{lk}(\UA{kl'}(\cdot))$ for all $k$ s.t. \mbox{$\min(\set{l,l'})\leq k \leq \max(\set{l,l'})$}.	\label{item b00}\\
	\item $\UA{ll'}(\cdot)$ and $\OA{ll'}(\cdot)$ are monotonic, i.e. $A_{l'}\subseteq B_{l'} \Rightarrow \UA{ll'}(A_{l'})\subseteq \UA{ll'}(B_{l'})$ and $A_{l'}\subseteq B_{l'} \Rightarrow \OA{ll'}(A_{l'})\subseteq \OA{ll'}(B_{l'})$. \label{item b0}\\
	\item For $l\leq l'$, both $\UA{ll'}(\cdot)$ and $\OA{ll'}(\cdot)$ are closed under union and intersection. \label{item closure}\\
	\item $l \leq l' \Rightarrow \UA{ll'}(\cdot) \equiv \OA{ll'}(\cdot)$ 
	\label{item b}\\
	\item $l\leq l' \Rightarrow \UA{l'l}(\UA{ll'}(A_{l'})) = \OA{l'l}(\UA{ll'}(A_{l'})) = A_{l'}$. Using \eqref{item b}, we additionally have $l\leq l' \Rightarrow \UA{l'l}(\OA{ll'}(A_{l'})) = \OA{l'l}(\OA{ll'}(A_{l'})) = A_{l'}$, i.e., when $l\leq l'$, the composition $\Gamma^*_{l'l} \circ \Gamma^*_{ll'}$ for $*\in \set{\uparrow,\downarrow}$ is the identity function.\\
	\item For all $x\in X$, $\Qa{l'}(x)\in A_{l'} \Rightarrow \Qa{l}(x)\in \OA{ll'}(A_{l'})$. Equivalently, for all $\xa'\in \Xa{l'}$, $\xa'\in A_{l'}\Rightarrow\Ra{ll'}(\xa')\in \OA{ll'}(A_{l'})$. \label{item e}\\
	\item For all $x\in X$, $\Qa{l}(x)\in \UA{ll'}(A_{l'}) \Rightarrow \Qa{l'}(x)\in A_{l'}$. Equivalently, for all $\xa\in \Xa{l}$, $\xa\in \UA{ll'}(A_{l'})\Rightarrow\Ra{l'l}(\xa)\in A_{l'}$. \label{item f}
\end{inparaenum}

Using (a)-(f), it immediately follows that
\begin{subequations}\label{eq dual}
	\begin{align}
		&l<l':	&&	\UA{ll'}(A_{l'}) \subseteq A_l \Rightarrow	A_{l'} \subseteq \OA{l'l}(A_l)	\label{eq dual a}\\
		&l>l':	&&	\UA{ll'}(A_{l'}) \subseteq A_l \Leftarrow	A_{l'} \subseteq \OA{l'l}(A_l)	\label{eq dual b}\\
		&\text{Always}:		&&	\UA{ll'}(A_{l'}) \supseteq A_l \Leftrightarrow	A_{l'} \supseteq \OA{l'l}(A_l) \label{eq dual c}
	\end{align}
\end{subequations}
where $A_l \subseteq \Xa{l}$. The implications in \eqref{eq dual a} and \eqref{eq dual b} are strict. 



The soundness and relative completeness of $\MLAR_{m}$ follows from \REFthm{thm:EagerReach}, 
if we can ensure that in every iteration of $\MLAR_{m}$ the set of states 
returned by $\EXPLORE_m$, for which the abstract transition relation is computed, 
is \emph{not smaller} than the set of states subsequently added to $\Aux$ by $\REACH_m$ in the next iteration. 
We obtain this result by the following series of lemmata.

We first observe that computing the under-approximation of the $m$-step cooperative predecessor w.r.t.\ the auxiliary system $\Aa{l}$ of a set $\Aux_L$ (as used in $\EXPLORE_m$) over-approximates the set obtained by computing the $m$-step cooperative predecessor w.r.t.\ the abstract system $\Sa{l}$ for a set $\Aux_l$ (as used in $\REACH_m$) if $\Aux_L$ over-approximates $\Aux_l$. 

\begin{lemma}\label{lem:containment}
Let $\Saset$ be a multi-layered abstract system satisfying Assumption~\ref{assumption:monotonic approximation}, and let $\Aux_l\subseteq \Xa{l}$ and $\Aux_L\subseteq \Xa{L}$ for some $l<L$ s.t.\ $\Aux_L \supseteq \OA{Ll}(\Aux_l)$. Then $\UA{lL}(\JUpre{l}(\Aux_L)) \supseteq \FUpre{l}(\Aux_l)$. Furthermore, for all $m>0$, it holds that $\UA{lL}(\JUpre{l}^m(\Aux_L)) \supseteq \FUpre{l}^m(\Aux_l)$, where $\FUpre{l}$ and $\FUpre{l}^m$ for $\Sa{l}$ are defined analogously to \eqref{eq:define upre} and \eqref{eq:def upre^r}, respectively.
\end{lemma}

\begin{proof}[Proof of \REFlem{lem:containment}]
	Let $\xa\in \FUpre{l}(\Aux_l)$, which by definition \eqref{eq:define upre} implies that $\Fa{l}(\xa) \cap \Aux_l \neq \emptyset$.
	Let $\ya = \OA{Ll}(\set{\xa})$. 
	Then by observing that   $\xa \subseteq \ya$, and using Assump.~\ref{assumption:monotonic approximation}, we have that $\Fa{l}(\xa)\subseteq \Ja{l}(\ya)$, which implies that $\Ja{l}(\ya) \cap \Aux_L \neq \emptyset$ (since $\Aux_l\subseteq \Aux_L$).
	Hence, $\ya\in \JUpre{l}(\Aux_L)$.
	Moreover, using \eqref{eq dual c} we have that $\xa\in\UA{lL}(\set{\ya})$ which leads to $\xa\in \UA{lL}(\JUpre{l}^1(\Aux_L))$.

       The second claim is proven by induction on $m$. The base case for $m=1$ is given by the first claim proven above. Now assume that $\UA{lL}(\JUpre{l}^m(\Aux_L)) \supseteq \FUpre{l}^m(\Aux_l)$ holds for some $m>0$. This together with \eqref{eq dual c} implies:
	\begin{equation} \label{eq:proof ineq a}
		\JUpre{l}^m(\Aux_L) \supseteq \OA{Ll}(\FUpre{l}^m(\Aux_l)).
	\end{equation}
	 Now note that by \eqref{eq:def upre^r}, we have $\JUpre{l}^{m+1}(\cdot) = \JUpre{l}(\JUpre{l}^m(\cdot)) \cup \JUpre{l}^m(\cdot)$ and it holds that
	 \begin{align*}
	 	&\UA{lL}(\JUpre{l}^{m+1}(\Aux_L))\\
	 	 =& \UA{lL}\left(\JUpre{l}\left(\JUpre{l}^m(\Aux_L)\right) \cup \JUpre{l}^m(\Aux_L)\right) \\
	 						=& \UA{lL}\left(\JUpre{l}\left(\JUpre{l}^m(\Aux_L)\right)\right) \cup \UA{lL}\left(\JUpre{l}^m(\Aux_L)\right) \numberthis \label{proof eq a0}\\
	 						\supseteq& \FUpre{l}\left(\FUpre{l}^m(\Aux_l)\right) \cup \FUpre{l}^m(\Aux_l) \numberthis \label{proof eq a}
	 						=\FUpre{l}^{m+1}(\Aux_l),
	 \end{align*}
	 where \eqref{proof eq a0} follows from (\ref{item closure}) and \eqref{proof eq a} follows by applying the first claim twice: (i) for the left side of the ``$\cup$'', by replacing $\Aux_l$ and $\Aux_L$ in the first claim of \REFlem{lem:containment} by $\FUpre{l}^m(\Aux_l)$ and $\JUpre{l}^m(\Aux_L)$ respectively, while noting that \eqref{eq:proof ineq a} gives the necessary pre-condition, and (ii) for the right side of the ``$\cup$''.
\end{proof}

\REFlem{lem:containment} can be used to show that $\EXPLORE_m$ constructs the transition function $\Fa{l}(\xa,\ua)$ for all $\xa$ which are in the winning state set computed by $\REACH_m$.

\begin{lemma}\label{lem:single iteration}
For all $l<L$, $\Aux\subseteq\Xa{1}$ and $C=(\Uci{},\Ua,\Gci{})$ returned by $\REACH_m(\UA{l1}(\Aux),\Oa{l},l)$, it holds that $\xa\in\Uci{}\setminus \UA{l1}(\Aux)$ implies $\xa\in W''$, where $W''$ is returned by the second line of $\EXPLORE_m(\Aux,l)$.
\end{lemma}

\begin{proof}[Proof of \REFlem{lem:single iteration}]
By assumption, we have $\xa\in B\setminus \UA{l1}(\Aux)$ i.e. (i) $\xa\in B$ and (ii) $\xa\notin \UA{l1}(\Aux)$. Then it follows from (i) that
 \begin{align*}
 		\xa\in B
		\Rightarrow \ \xa \in \FCpre{l}^m(\UA{l1}(\Aux))
		\Rightarrow \ \xa\in \FUpre{l}^m(\UA{l1}(\Aux)).	
	\end{align*}
	Consider the inequality $\OA{L1}(\Aux) \supseteq \OA{Ll}(\UA{l1}(\Aux))$ which can be verified from the properties (a)-(f). Then \REFlem{lem:containment} and (\ref{item f}) give
	\begin{align*}
		\xa\in \UA{lL}(\JUpre{l}^m(\OA{L1}(\Aux)))
		\Rightarrow \Ra{Ll}(\xa)\in \JUpre{l}^m(\OA{L1}(\Aux)).		
	\end{align*}
	 Now (ii) and (\ref{item f}) gives 
	 \begin{align*}
		\Ra{Ll}(\xa)\notin \UA{Ll}(\UA{l1}(\Aux))	
		\Rightarrow \ \Ra{Ll}(\xa) \notin \UA{L1}(\Aux).	
	 \end{align*}
	 Combining the last two observations with (\ref{item e}) and (\ref{item b0}) we get
	 \begin{align*}
	 	\Ra{Ll}(\xa)&\in \JUpre{l}^m(\OA{L1}(\Aux)) \setminus \UA{L1}(\Aux)
	 					 = W'\\ 
	 	\Rightarrow \ \xa &\in \OA{lL}(W')	
	 	\Rightarrow \ \xa \in \UA{lL}(W')	
	 					 =W'', 				
	 \end{align*}
	 which proves the claim.
\end{proof}

With \REFlem{lem:single iteration}, soundness and relative completeness of $\MLAR_{m}$ directly follows from \REFthm{thm:EagerReach}, as shown in the following.
%
%
	We build the proof on top of \REFthm{thm:EagerReach}. 
	We prove two things: that both algorithms terminate after the same depth of recursion $D$, and that the overall controller domain that we get from $\MLR_m$ is same as the one that we get from $\MLAR_{m}$, i.e. $\cup_{d\in[1;D]}\underline{B}^d = \cup_{d\in[1;D]}B^d$, where $\underline{B}^d$ and $B^d$ are the controller domains obtained in depth $d$ of the algorithms $\MLR_m$ and $\MLAR_{m}$ respectively. (We actually prove a stronger statement: for all $d\in[1;D]$, $\underline{B}^d = B^d$.) Then, since $\MLR_m$ is sound and complete w.r.t. $\REACH_\infty$, hence $\MLAR_{m}$ will also be sound and complete w.r.t. $\REACH_\infty$.
	
	The ``$\supseteq$'' direction of the second proof is trivial and is based on two simple observations: \begin{inparaenum}[(a)]\item the amount of information of the abstract transition systems $\Saset$ which is available to $\MLAR_{m}$ is never greater than the same available to $\MLR_m$; \item whenever $\MLAR_{m}$ invokes $\EXPLORE_m$ for computing transitions for some set of abstract states, $\EXPLORE_m$ returns the \emph{full information} of the outgoing transitions for those states to $\MLAR_{m}$. The second part is crucial, as partial information of outgoing transitions might possibly lead to false positive states in the controller domain. \end{inparaenum}  Combining these two arguments, we have that for all $d\in [1;D]$ $\underline{B}^d \supseteq B^d$. (We are yet to show that the maximum recursion depth is $D$ for both the algorithms $\MLR_m$ and $\MLAR_{m}$.)
	
	The other direction will be proven by induction on the depth of the recursive calls of the two algorithms. Let $\underline{l}_d$ and $l_d$ be the corresponding layers in depth $d$ of algorithm $\MLAR_{m}$ and $\MLR_m$ respectively.  It is clear that $\underline{B}^1 = B^1$ and $\underline{l}_1 = l_1 = L$ (induction base) since we start with full abstract transition system for layer $L$ in both cases. Let us assume that for some depth $d$, $\underline{B}^{d'} = B^{d'}$ and $\underline{l}_{d'} = l_{d'}$ holds true for all $d'\leq d$ (induction hypothesis). 
	Now in $\MLAR_{m}$, the check in Line~\ref{alg:checkFP} of $\MLAR_{m}$ is fulfilled iff the corresponding check in Line~15 of $\MLR_m$ (i.e. \cite[Alg.~1]{hsu2018multi}) is fulfilled. This means that $\underline{l}_{d+1} = l_{d+1}$. This shows by induction that \begin{inparaenum}[(a)] \item the maximum depth of recursion in $\MLAR_{m}$ and $\MLR_m$ are the same (call it $D$), and \item the concerned layer in each recursive call is the same for both algorithms. \end{inparaenum}
	
	Now in the beginning of depth $d+1$, we have that $\Aux = \cup_{d'\leq d} \UA{1l_{d'}}\underline{B}^{d'} = \cup_{d'\leq d} \UA{1l_{d'}}B^{d'}$. From now on, let's call $l_{d+1} = l$ for simpler notation. Let $\xa\in\Xa{l}$ be a state which was added in depth $d+1$ in the controller domain $\underline{B}^{d+1}$ for the first time, i.e. \begin{inparaenum}[(a)]\item $\xa \in \underline{B}^{d+1}$, and \label{part a} \item $\xa \in \UA{l1}(\Aux)$ \label{part b} \end{inparaenum}. Then by \REFlem{lem:single iteration} we have that $\xa \in W''$.
	
	Since $\EXPLORE_m$ also computes all the outgoing transitions from the states in $W''$ (Line~\ref{alg:partial abs} in \REFalg{alg:EXPLORE subroutine}), hence full information of the outgoing transitions of all the states which are added in $\underline{B}^{d+1}$ will be available to the $\MLAR_{m}$ algorithm in depth $d+1$. In other words given $\xa\in \Xa{l}$, if there is an $m$-step controllable path from $\xa$ to $\Aux$ in $\MLR_m$, there will be an $m$-step controllable path in $\MLAR_{m}$ as well. Hence $\xa$ will be added in $B^{d+1}$ as well. This proves that for all $d\in [1;D]$ $\underline{B}^d \subseteq B^d$.

\end{document}